\documentclass[twoside,11pt,final]{entics} 
\usepackage{amsmath}

\usepackage{enticsmacro}
\usepackage{graphicx}
\usepackage{tikz-cd}
\usepackage[all]{xy}
\usepackage{enumitem}
\usepackage{scalerel}
\usepackage{wrapfig}
\usepackage{stmaryrd}
\def\conf{MFPS 2025} 	%%Fill in the acronym for your conference (with year)
\volume{5}			%Fill in the ENTICS volume number here
\def\volu{5}			% and here
\def\papno{18}			%Fill in your paper number here
\def\mydoi{16816}
\def\lastname{Plummer and C\^irstea} %Lastnames appear in the running header 
\def\runtitle{Traces via Strategies} %Short title appears in the running header on even pages. 
\def\copynames{B. Plummer, C. C\^irstea}    %% Fill in the first initial and last name of the authors
\newtheorem{counterexample}{Counterexample}

\newcommand{\bb}{\mathbb}
\newcommand{\mc}{\mathcal}
\newcommand{\mt}{\textnormal}
\newcommand{\tts}{\textstyle}
\newcommand{\ds}{\displaystyle}
\newcommand{\pown}{P^{\raisebox{1pt}{\scaleto{+}{4pt}}\hspace{-1pt}}}

\input{widebar}

\begin{document}
%%%Note the beginning and end of the frontmatter section that starts here%%%%%
\begin{frontmatter}
  \title{Traces via Strategies in Two-Player Games}
  %%%%Now the author(s) names(s)%%%%%
  \author{Benjamin Plummer\thanksref{a}\thanksref{email}\thanksref{main}}	%%Note NO SPACE between 
  \author{Corina C\^irstea\thanksref{a}\thanksref{coemail}\thanksref{co}}		%last name and \thanksref{...} 
    %%%Next come the addresses%%%%
   \thanks[email]{Email: \href{bjp1g19@soton.ac.uk} {\texttt{\normalshape
        bjp1g19@soton.ac.uk}}} 
   %%%Note: if both authors share same institution, only list the address once, after the second 
   %%%author. 
   %%%There also is a link from the first author to the co-author's address to show how to list 
   %%%affiliations to more than one institution, when needed. 
  \address[a]{School of Electronics and Computer Science\\University of Southampton\\ United Kingdom} 
  \thanks[coemail]{Email:  \href{cc2@ecs.soton.ac.uk} 
  {\texttt{\normalshape
        cc2@ecs.soton.ac.uk}}}
    \thanks[main]{The first author would like to thank Alexandre Goy for enlightening discussions about the contents of his thesis.}
  \thanks[co]{C\^{\i}rstea was funded by a Leverhulme Trust Research Project Grant (RPG-2020-232).}
  %\thanks[email]{}
\begin{abstract} 
Traces form a coarse notion of semantic equivalence
between states of a process,
and have been studied coalgebraically for various
types of system.
We instantiate the finitary coalgebraic trace semantics framework of Hasuo et al.
for controller-versus-environment games,
encompassing both nondeterministic and probabilistic environments.
Although our choice of monads is guided by the constraints of this abstract framework,
they enable us to recover
familiar game-theoretic concepts.
Concretely,
we show that in these games, each element in the trace map corresponds
to a collection (a subset or distribution) of plays the controller can force.
Furthermore, each element can be seen as the outcome
of following a controller strategy.
Our results are parametrised by a weak distributive
law, which
computes what the controller can force in a single step.

\end{abstract}
\end{frontmatter}

\section{Introduction}

The problem of program synthesis can be phrased using a two-player game between
a controller and its environment: where the controller must achieve some
linear-time property irrespective of environment choices.
We make steps towards a \textit{coalgebraic} framework for game-based synthesis,
which allows us to uniformly treat modelling the environment as
nondeterministic or stochastic.

The process-theoretic notion of \textit{trace} is well studied in program semantics. A trace arises from a sequence of
choices a system can make, and records the observable behaviour of the resulting execution.
\emph{Plays} in two-player games are more general: players make interleaved choices and
have opposing objectives.
We work with an
elegant coalgebraic representation of two-player games, where each player is modelled
%modularly 
with a monad.
Combining the controller and environment monads with a weak distributive law,
we obtain a monad for the composite system, to which the general
coalgebraic theory of finite traces \cite{hasuo_generic_2007} is shown to apply.
%Using this coalgebraic machinery, we import traces into this richer setting
We then prove a close connection between controller strategies and traces,
phrasing game-theoretic notions like plays and strategies categorically along
the way.
Our parametric approach handles non-deterministic and probabilistic
environments uniformly, with two different monads.

We focus on two-player games where the observable outcome of a completed play consists
of a finite sequence of basic observations, each arising after a simple
controller-environment interaction.
To illustrate, consider the game on the following page, made up of controller states
(boxes), controller transitions (solid arrows),
environment states (black dots),
and environment transitions (dashed arrows).
A single observation from
$A=\{a,b,c,d,e\}$ is output after each transition
(when a play has not terminated), and an observation from $B=\{\checkmark\}$
is output when a play terminates.
We have chosen to model the environment as nondeterministic in this example.

%\begin{minipage}{0.4\textwidth}

\begin{wrapfigure}{r}{0cm}
    \hspace{-2em}
\begin{tikzpicture}
    \node[draw] at (0,0) (n1) {};
    \node[circle,inner sep=1pt,fill] at (1,0.5) (n2) {};
    \node[circle,inner sep=1pt,fill]at (1,-0.5) (n3) {};
    \node[draw] at (2,0.5) (n4) {};
    \node[circle,inner sep=1pt,fill] at (2,1.5) (n5) {};
    \node[draw] at (2,-0.5) (n6) {};
        
    \node[draw] at (4,0.5) (n8) {};
    \node[circle,inner sep=1pt,fill] at (3, 0.5) (n9) {};
    \node[circle,inner sep=1pt,fill] at (3, -0.5) (n10) {};
    \node[draw] at (4,-0.5) (n11) {};

    \node[circle,inner sep=1pt,fill] at (4.7,0.5) (e1) {};
    \node[] at (5.4,0.5) (e2) {};

    \node[circle,inner sep=1pt,fill] at (2,-1.15) (e3) {};
    \node[inner sep=0pt] at (2,-1.8) (e4) {};

    \draw[-stealth] (n1) -- (n2);
    \draw[-stealth] (n1) -- (n3);
    \draw[-stealth, dashed] (n2) -- (n4);
    \draw[-stealth] (n4) edge[bend right] (n5);
    \draw[-stealth, dashed] (n5) edge[bend right] (n4);
    \draw[-stealth, dashed] (n3) -- (n6);
    
    \draw[-stealth, dashed] (n3) -- (n4);
    \draw[-stealth] (n4) -- (n9);
    \draw[-stealth, dashed] (n9) -- (n8);
    \draw[-stealth] (n6) edge[bend left] (n10);
    \draw[-stealth, dashed] (n10) edge[bend left] (n6);
    \draw[-stealth, dashed] (n10) -- (n11);

    \draw[-stealth, dashed] (e1) -- (e2) {};
    \draw[-stealth] (n8) -- (e1) {};

    \node at (1.3, 0.65) {$a$};
    \node at (1.55, 1.3) {$b,c$};
    \node at (3.4, 0.7) {$d$};
    \node at (5.5, 0.55) {$\checkmark$};
    \node at (1.3, 0) {$a$};
    \node at (1.3, -0.7) {$b$};
    \node at (2.5, -0.8) {$e$};
    \node at (3.5, -0.35) {$e$};
    \node at (2.05,-1.9) {$\checkmark$};

    \node at (-0.4,0) {$x_0$};

    \draw[-stealth] (n6) -- (e3);
    \draw[-stealth, dashed] (e3) -- (e4);
\end{tikzpicture}
\end{wrapfigure}

A one-step controller-environment interaction is viewed as a single transition.
First, the controller chooses an environment state,
then the environment picks a pair of observation from $A$ and controller state,
or to terminate with an observation from $B$.
Each completed play from a state can be associated with a \textit{trace}:
a sequence of observations from $A$ ending with termination and a final observation from $B$.
Our main theorem shows that the trace semantics of a state, obtained by instantiating the coalgebraic theory of finite traces to our composite monad, is the set of subsets of terminating traces which
the controller can force.
For example, from $x_0$, the subsets of traces the controller can force include
$\{b\checkmark,ad\checkmark\}$, $\{abd\checkmark,acd\checkmark\}$ and
$\{b\checkmark, acd\checkmark,abcd\checkmark,abbcd\checkmark,abbbcd\checkmark,abbbbd\checkmark\}$, which correspond
to history-dependent controller strategies which force them.
Notice that some strategies are not {\it finitely completing}, i.e. do not force a set of finite traces.
When the environment is modelled probabilistically,
the subsets of enforceable traces become distributions of traces
that the controller can force.

In coalgebraic trace semantics,
the type of system under consideration
is modelled with a monad (describing the computation type, e.g. non-deterministic or probabilistic)
composed with an endofunctor (describing the type of observations made along a trace).
We begin, following work on alternating automata
\cite{goy_compositionality_2021,goy_powerset-like_2021},
by identifying a composite monad built from
the weak distributive law of powerset over itself
as a suitable model for controller-environment interactions in two-player games.
We then refine our component monads -
to satisfy the requirements of the finite trace semantics theorem in \cite{hasuo_generic_2007} -
to full powerset for controller choices
and finite, non-empty powerset for environment choices.
In the process of examining the conditions of \cite{hasuo_generic_2007},
we reveal two mistakes in the literature:
namely that composition in the Kleisli category of the monad featured in \cite{bonchi_convexity_2022}
is not left-strict;
and that the monad in \cite{jacobs_traces_2009}
is not commutative.
Section~\ref{sect:traces_and_execs} proceeds by proving the conditions required 
to obtain a trace map, and then introduces \textit{execution maps} (where executions are the coalgebraic counterpart of game plays)
as a special kind of trace map.
The required conditions also hold when using full powerset
for controller choices, and the finite distribution monad for environment choices,
giving us a second example where the environment evolves probabilistically.
We finish the section by showing that the trace map factors through the execution map.

A standard definition of a player strategy in a two-player game,
is a function from \textit{partial plays} ending in a state controlled by that player,
to a valid move for that player. 
Such strategies need only be defined over partial plays
which they can force, so are partial functions.
In Section~\ref{sect:strategies},
we give an equivalent definition of a controller strategy as a chain of maps in the
Kleisli category of the monad modelling the environment.
Then, in Section~\ref{sect:executions_via_strategies},
we characterise the trace map
in terms of collections (subsets or distributions) of finite traces that a controller strategy can force.
%can be forced by some controller strategy.
We show that a
collection of traces is in the trace map at a state $x$ if and only if,
there exists a strategy
from $x$
forcing exactly that collection.
We see this as a direct result
of the underlying weak distributive law computing
one-step outcomes of strategies.

%Our main result in Section~\ref{sect:executions_via_strategies},
%confirms that we have found the right monad
%to use in controller synthesis.
%The fact that traces, and strategies, arise by coinduction,
%paves the way for generic coinductive algorithms for
%quantitative synthesis,
%whereby controller strategies can be computed by a greatest fixpoint. 
Our general goal is to develop theory to underpin
a general coalgebraic synthesis tool.
The main result, Theorem~\ref{sect:traces_and_execs},
confirms we have found a suitable monad for representing
two-player games
for controller synthesis.
The practical contribution
is that outcomes of {\it finitely completing} strategies in a game,
are the finite trace semantics
of said game, so can be computed with a least fixed-point
computation.
Computing (finitely completing) controllers is thus reduced to an inductive computation.
Our choice to model a non-deterministic environment using the \emph{non-empty} finite powerset, also
makes sense from a practical standpoint: a real environment can never deadlock, and it is common practice to restrict the possible environment inputs to be finite\footnote{A similar assumption is made in geometric logic \cite{Vickers},
where there is arbitrary disjunction and finite conjunction.} to make the synthesis problem decidable.
The composite monad will give \textit{convex} sets of collections (subsets or distributions),
meaning that sets of subsets are closed under binary union,
and sets of distributions are closed under convex combination.
In finitely branching games, with a non-deterministic environment,
this yields a greatest choice at each controller state, the union of all the other choices.
%This should allow strategy synthesis to be performed via a greatest fixpoint computation.
This can allow strategy synthesis by approximating down from the largest strategy, yielding the most permissive strategy.
For probabilistic enrivonments,
convexity allows 
\textit{randomised strategies},
which are standard in the theory of Markov Decision
Processes \cite{mdps}, to be accounted for.

Finally, we comment on the role of category theory in this paper.
Our main result, equating trace semantics and outcomes of strategies,
cannot be known explicitly, as
there is no definition of trace semantics of two-player
games in the literature.
What our coalgebraic approach gives us,
is a ready-made definition of finite trace semantics:
it is a fact (established in Section~\ref{sect:traces_and_execs})
that there is a final coalgebra in the category
where our games live.
Thus, the role of coalgebra, and more generally category theory,
is pivotal in our work.
In Section~\ref{sect:strategies},
we give a category-theoretic definition of strategies
and plays, which is not only enlightening,
but also provides a clear road map for the proof
in Section~\ref{sect:executions_via_strategies}.

\subsection{Related Work}

A composite monad which we do not consider,
but is conceptually very similar,
is full powerset combined with full finite powerset.
This monad is a special case of the monad
considered in \cite{bonchi_convexity_2022},
which combines full powerset
with a \textit{multiset monad over a semifield},
when we choose the Boolean semifield.
However, as we will show,
this monad is inadequate for trace semantics given
by Theorem \ref{theorem:traces}.
Similarly,
the monad considered in \cite{jacobs_coalgebraic_2008}
could instantiate
to non-empty powerset combined with finite powerset.
This is also not suitable for Theorem \ref{theorem:traces},
however is adapted to work in \cite{jacobs_coalgebraic_2008}.
In doing so, they require an additional assumption
which, when expressed in terms of games,
amounts to the controller always being able to force
the environment to deadlock immediately.
This would be unsatisfactory in %our strategic characterisation of traces in 
controller-versus-environment games, as
the controller would always have a (trivial) winning strategy.

When looking for monads on the double powerset functor,
the \textit{neighbourhood monad} $\mc{N}$ comes to mind
(generated from the contravariant powerset being dual
adjoint to itself).
The multiplication $\mu^\mc{N}$
does not model games,
for example it only returns a non-empty set
when given a subset containing  an upwards closed set.
Similarly, double covariant powerset on a function
always returns upwards closed sets.
To fix these oddities, a common choice
for work on game logic \cite{hansen-weak-complete5,hansen-parity-games}
is to use the \textit{monotone neighbourhood monad} $\mc{M}$;
it has been suggested in \cite{kojima} to use
$\mc{M}$ to give a path-based semantics for a coalgebraic CTL.
This monad is generated by restricting
the dual adjunction between ${\bf Set}$ and itself to one
between ${\bf Set}$ and ${\bf Poset}$ (see e.g. \cite{jacobs_recipe_2015})
- and was derived independently in \cite{bertrand_coalgebraic_2018}
in the context of alternating automata, using
a distributive law.
We view the monotone neighbourhood monad as ``fixing''
the neighbourhood monad because the contravariant nature
is tamed: on upward closed sets $PP(f)$ agrees with $\mc{N}(f)$
and the multiplication agrees with the $\exists$ $\forall$
behaviour of picking a strategy\footnote{The monotone neighbourhood monad is
not a submonad of $\widetilde{PP}$,
because the units do not agree.}.
This monad is also not suitable for the
assumptions required in Theorem \ref{theorem:traces}, because
its Kleisli category is not enriched in $\omega$-cpos
(see Counterexample~\ref{counter:omega-cpo}).
It may also be possible, analogous to the approach taken
in this paper, to restrict to upwards closed
sets of {\it finite} subsets, we leave investigating
this to future work.

In \cite{BOWLER201823},
the authors consider a non-deterministic
programming language,
where game-like behaviour arises from
I/O interaction.
The key difference with our work,
is that the input to the program
is part of the (branching-time) behaviour -
captured with a functor rather
than a monad.
The non-determinism of the program,
is treated with the powerset monad,
thus the trace semantics is a set of possible computations
(rather than a set of subsets, like in ours).

Finally, while we have chosen to work with the Kleisli approach to finite trace semantics
from \cite{hasuo_generic_2007},
it is worth mentioning other approaches and variants.
Our games are generative (of the shape $TF$),
which rules out the Eilenberg-Moore
approach (which treats systems of type $GT$).
In \cite{jacobs_trace_2012},
the authors give a way of casting the Kleisli approach to the Eilenberg-Moore one,
using a natural transformation $TF\rightarrow GT$.
We find it unlikely that a suitable functor $G$ and
natural transformation exist in our case.
Work in \cite{frank_coalgebraic_2022} provides a
different set of assumptions to obtain finite trace semantics,
in regards to our work the only assumption that makes a difference
is they do not assume a zero object in the Kleisli category of the monad,
however they do still require the stronger condition of left-strictness
(see \cite[Lemma 3.5]{hasuo_generic_2007}).
It may be possible to phrase our work
in terms of graded monads \cite{milius-graded}.
Although, we currently rely on structure in
the Kleisli category of our monad (e.g. $\omega$-cpo enrichedness and
the existence of a certain limit)
which
is afforded to us by the approach in \cite{hasuo_generic_2007,Jacobs_2016},
hence leave a treatment in terms of graded monads to future work.

\section{Outline}

We now sketch the categorical ideas of our approach. % in this section.
Assume we model the controller with the full covariant powerset monad
$P:{\bf Set}\rightarrow{\bf Set}$,
and the environment with some monad $T:{\bf Set}\to{\bf Set}$.
Suppose we have some way to combine these into a monad
$\widetilde{PT}:{\bf Set}\rightarrow{\bf Set}$, which fits
the framework for finite trace semantics in \cite{hasuo_generic_2007}.
We model two-player games as $\widetilde{PT}H$ coalgebras,
with a functor $H:{\bf Set}\rightarrow{\bf Set}$
describing the possible observations after one interaction.
We will take $H(X)=B+A\times X$, focussing on games whose plays can terminate
with an observation $b\in B$, or proceed to a new state $x\in X$ with an
observation $a\in A$, following a controller-then-environment
choice of moves.
Instantiating \cite{hasuo_generic_2007},
we get a trace map ${\sf tr}_c:X\rightarrow \widetilde{PT}(A^*B)$
for each coalgebra $c:X\rightarrow\widetilde{PT}H(X)$,
where $A^*B$ is the initial $H$-algebra.
We answer two questions:
\begin{enumerate}
    \item There are no established notions of trace semantics in games,
        so what do the contents of ${\sf tr}_c(x)$ correspond to?
    \item How do strategies fit into the categorical picture?
\end{enumerate}

We will summarise our answers shortly.
First, fix
a coalgebra $c:X\rightarrow \widetilde{PT}H(X)$,
that we view as a game.
The states $x\in X$ correspond to controller states,
whereas elements $U\in c(x)\subseteq T(AX+B)$ are collections
of observations from one-step plays which the controller can force from $x$.
These elements can be thought of as environment states.
Similarly, iterating (in the Kleisli category of the monad $\widetilde{PT}$) the coalgebra
gives a map $X\to\widetilde{PT}H^n(X)$,
assigning each state to the set of collections of observations of $n$-step 
plays which the controller can force.

To answer the first question,
we follow \cite{hasuo_generic_2007},
and unpack the generic construction of the trace map.
The colimit of the initial sequence of the endofunctor $H$
gives the initial $H$-algebra $A^*B$.
The elements of this algebra are the possible observable outcomes of plays,
each given by some $\kappa_n(a_1 \ldots a_{n-1} b)$ with $n \ge 1$:
\begin{center}
\tikz[overlay]{
        \node[draw] at (1,-0.5) {{\bf Set}};
    }
\begin{tikzcd}
    0 \ar[r, "!"] \ar[rrrd, bend right=10, "\kappa_0"] &[2em] H(0) \ar[r, "H(!)"] \ar[rrd, bend right=5, "\kappa_1"]
        &[2em] H^2(0) \ar[r, "H^2(!)"] \ar[dr, "\kappa_2"] &[2em]
        \cdots
        \\[-1em]
        &&& A^*B
\end{tikzcd}
\end{center}
The trace map from \cite{hasuo_generic_2007}
arises from the observation that the initial $H$-algebra is a final
$\widebar{H}$-coalgebra (where $\widebar{H}:{\bf Kl}(\widetilde{PT}) \to {\bf Kl}(\widetilde{PT})$
is the extension of $H:{\bf Set}\to{\bf Set}$).
\begin{center}
\tikz[overlay]{
        \node[draw] at (-0.7,0) {{\bf Kl}$(\widetilde{PT})$};
    }
\begin{tikzcd}
    X \ar[r, "c"] \ar[d, "!"] &[2em] \widebar{H}(X) \ar[r, "\widebar{H}(c)"] \ar[d, "\widebar{H}(!)"]
    &[2em] \widebar{H^2}(X)  \ar[d, "\widebar{H^2}(!)"] \ar[r, "\widebar{H^2}(c)"]&[2em] \cdots
    \\[-0.5em]
    0   &[2em] \widebar{H}(0) \ar[l, "!"']
        &[2em] \widebar{H^2}(0)  \ar[l, "\widebar{H}(!)"'] &[2em]
        \cdots \ar[l, "\widebar{H^2}(!)"']
                \\[-0.5em]
        &&& A^*B \ar[ulll, bend left=10, "\pi_0"'] \ar[ull, bend left=5, "\pi_1"'] \ar[ul, "\pi_2"']
\end{tikzcd}
\end{center}
Specifically, iterating $c$ up to some depth and then projecting into the
final sequence of $\widebar H$ yields a cone over this final sequence,
and the map ${\sf tr}_c:X\rightarrow A^*B$ arises from the limiting property
of $A^*B$ in ${\bf Kl}(\widetilde{PT})$.
For $x \in X$, each element of ${\sf tr}_c(x)$ is a collection of observations
of completed plays, which the controller can force.
This provides an answer to the first question.

For the second question,
note that strategies (in the standard sense) resolve controller choices,
so should exist in the Kleisli category of the monad $T$ used to model
environment choices. 
In Section~\ref{sect:strategies}, roughly speaking,
we capture strategies using a family of maps $\sigma_{n+1}:\mt{Im}(\sigma_n)\rightarrow \widehat{H^{n+1}}(X)$ for each $n\in\omega$, depicted below.

\begin{center}
\tikz[overlay]{
        \node[draw] at (-0.7,0) {{\bf Kl}$(T)$};
    }
\begin{tikzcd}
      & X  & \widehat{H}(X)  & \widehat{H^2}(X) & \dots
     \\[-0.5em]
      1 \ar[ur, "\sigma_0"] \ar[r, dashed]
      & \mt{Im}(\sigma_0) \ar[ur, "\sigma_1"] \ar[r, dashed] \ar[u, >->]
      & \mt{Im}(\sigma_1) \ar[ur, "\sigma_2"] \ar[r, dashed] \ar[u, >->]
      & \mt{Im}(\sigma_2) \ar[ur, "\sigma_3"] \ar[r, dashed] \ar[u, >->] & \cdots
\end{tikzcd}
\end{center}
Here, $\widehat{H}$ extends %\marginpar{\color{blue}I've been using extends terminology,but would happily change to lifts} 
$H$ to ${\bf Kl}(T)$. The map $\sigma_0$ picks the initial state of the play, whereas the maps $\sigma_{n+1}$ pick subsequent controller moves from states that can be reached via $\sigma_0, \ldots,\sigma_n$, according to the allowed moves in those states (as specified by $(X,c)$). Just like $\widebar{H^n}(X)$, the elements of $\widehat{H^n}(X)$ are either completed plays of length up to $n$ or incomplete plays of length exactly $n$. By composing along the bottom of the above diagram we obtain maps ${\sf plays}_n^\sigma:1\rightarrow \widehat{H^n}(X)$, which give the set of complete plays (of length less than or equal to $n$) and partial plays (of length $n$) which the strategy $\sigma$ can force.
We show in Section~\ref{sect:strategies} that we can lift these maps
into ${\bf Kl}(\widetilde{PT})$, and they again form a cone over the final sequence.

\begin{center}
\tikz[overlay]{
        \node[draw] at (-0.7,0) {{\bf Kl}$(\widetilde{PT})$};
    }
\begin{tikzcd}
    1 \ar[d, "{\sf plays}^\sigma_0"]
    \ar[dr, bend left=5, "{\sf plays}^\sigma_1"{yshift=-3pt}]
    \ar[drr, bend left=10, "{\sf plays}^\sigma_2"]
    \\[-0.5em]
    X \ar[d, "!"] &[2em] \widebar{H}(X) \ar[d, "\widebar{H}(!)"]
    &[2em] \widebar{H^2}(X)  \ar[d, "\widebar{H^2}(!)"] &[2em] \cdots
    \\[-0.5em]
    0   &[2em] \widebar{H}(0) \ar[l, "!"']
        &[2em] \widebar{H^2}(0)  \ar[l, "\widebar{H}(!)"'] &[2em]
        \cdots \ar[l, "\widebar{H^2}(!)"']
\end{tikzcd}
\end{center}
This establishes a unique mediating map $1\rightarrow A^*B$
in ${\bf Kl}(\widetilde{PT})$ for each strategy $\sigma$,
which we see as the outcome of $\sigma$:
which finite completed traces $\sigma$ can force.
We complete the answer to our second question in
Section~\ref{sect:executions_via_strategies},
by proving that the trace semantics ${\sf tr}_c(x)$ at a state $x$
coincides with the union of all the finite observable outcomes of strategies from $x$.
This gives us that there is a set $U\in{\sf tr}_c(x)$
if and only if, there is a finitely completing strategy which forces
a set of completed plays underlying the traces in $U$.

\section{Preliminaries}
\label{sect:prelim}

\subsection{Two-player Games}
\label{sect:two_player_games}

We start by introducing the kind of two-player game we are concerned with,
and the subsequent game-theoretic concepts.
These are represented coalgebraically in the remainder of the paper.

Fix two disjoint sets $A$ and $B$,
representing the continuing observations and terminating
observations respectively.
Our two-player games are player over {\it bipartite game graphs}
$(X,Y,E_1,E_2)$, where
$X$ and $Y$ are disjoint sets of controller and environment states respectively,
and $E_1\subseteq X\times Y$ and $E_2\subseteq Y\times (B + A\times X)$
are the controller and environment edge relations respectively.
For technical reasons discussed in Section~\ref{sect:traces_and_execs},
we put two restrictions on the environment's edge relation $E_2$:
that it is \textit{image-finite} and \textit{left-total}.
Image-finiteness means
$\{u\in B+A\times X\mid E_2(y,u)\}$ is finite for each environment state $y\in Y$.
$E_2$ being left-total means for each environment state $y\in Y$, there exists
some $u\in B+A\times X$ such that $E_2(y,u)$.

A \textit{partial play} over a game graph % from a controller state $x_0\in X$,
is an element $x_0a_1x_1\dots a_nx_n\in (XA)^*X$
such that there is an environment state $y\in Y$ with
$E_1(x_i,y)$ and $E_2(y,(a_{i+1},x_{i+1}))$, for all $0\le i<n$.
A \textit{completed play} is an element $\rho xb\in (XA)^*XB$,
such that $\rho x$ is a partial play and there
exists an environment state $y\in Y$ with $E_1(x,y)$ and $E_2(y,b)$.

Let $\sigma:(XA)^*X\rightarrow Y$ be a \textit{partial function} (i.e.
not defined over the entire domain),
which respects our game: $E_1(x,\sigma(\rho x))$ for every $\rho x\in (XA)^*X$
which $\sigma$ is defined over.
We say that a partial play $x_0a_1x_1\dots a_nx_n\in (XA)^*X$
\textit{conforms to} $\sigma$, if for all $0\le i < n$:
$\sigma$ is defined over $x_0a_1x_1\dots a_ix_i$
and $E_2(\sigma(x_0a_1x_1\dots a_ix_i),(a_{i+1},x_{i+1}))$.
Similarly, a completed play $\rho xb\in (XA)^*XB$,
made up of partial play $\rho x$ and a terminating observation $b$,
\textit{conforms to} $\sigma$,
when $\rho x$ conforms to $\sigma$ and $\sigma(\rho x)$
is defined with $E_2(\sigma(\rho x),b)$.
We call a pair $(x,\sigma)$
a \textit{pointed standard strategy},
precisely when $\sigma$ is defined {\bf exactly} over the
$\sigma$-conform partial plays which start in $x$.
An $n${\it-step partial outcome} of a pointed standard strategy $(x,\sigma)$,
denoted ${\sf plays}^\sigma_n(x)$,
is the set of all partial plays of length $n$ (elements of $(XA)^nX$)
and of complete plays of length less than $n$ (elements of $(XA)^{<n}XB$)
which conform to $\sigma$.
The {\it completed outcome} ${\sf plays}^\sigma(x)\subseteq (XA)^*XB$
of a pointed standard strategy $(x,\sigma)$,
is the set of all completed plays which conform to $\sigma$.

The objective of the game,
from the point of view of the controller,
is to force a completed play within a set of ``good" outcomes,
where ``goodness" is a property of the \emph{trace} $a_1 \ldots a_n b$
underlying a completed play $x_0a_1x_1\dots a_nx_nb$.
The reader may have noticed that our notion of play differs
slightly from the standard one \cite{Mazala2002},
in that it does not record environment states.
This is not an issue, precisely because the environment states visited
along a play have no impact on whether the resulting outcome is good or not;
and as a result, strategies can not benefit
from recording environment states in the play history.
The controller states also do not affect the goodness of an outcome,
but they need to be explicit because the strategy may depend on them.

\subsection{Markov Decision Processes}

\label{sect:mdps}

We also consider a probabilistic variant of these games,
which are essentially \textit{Markov decision processes} (MDPs).
These games are still four-tuples $(X,Y,E_1,E_2)$,
but now $E_2$ is a function $Y\times (B+A\times X)\rightarrow[0,1]$
such that $E_2(y):B+A\times X\to [0,1]$
is a finitely supported probability distribution.

Here, a \emph{partial play} is an element $x_0a_1x_1\dots a_nx_n\in (XA)^*X$
such that for all $0\le i<n$, we have that there
exists some $y\in Y$ with $E_1(x_i,y)$ and $E_2(y,(a_{i+1},x_{i+1}))>0$.
A \emph{completed play} is a sequence $\rho xb\in(XA)^*XB$, such that $\rho x$
is a partial play, and there exists a $y\in Y$
with $E_1(x,y)$ and $E_2(y,b)>0$.

Let $\sigma:(XA)^*X\rightarrow Y$ be a partial function,
which respects our MDP: $E_1(x,\sigma(\rho x))$ for every $\rho x\in (XA)^*X$
which $\sigma$ is defined over.
A partial play $x_0a_1x_1\dots a_nx_n$ \emph{conforms to $\sigma$},
if for all $0\le i <n$ we have
$E_2(\sigma(x_0a_0x_1\dots a_ix_i),(a_{i+1},x_{i+1}))>0$.
Again, a completed play $\rho xb$ {\it conforms to }$\sigma$,
if $\rho x$ is a partial play which conforms to $\sigma$,
and $E_2(\sigma(\rho x),b)>0$.
We call a pair $(x,\sigma)$ a {\it pointed standard strategy},
precisely when $\sigma$ is defined {\bf exactly} over
the $\sigma$-conform partial plays which start in $x$.
The $n$-step partial outcome
and the $n$-step completed outcome,
then become distributions over partial and completed plays.
We leave the definitions to the reader.
%\[
%    {\sf plays}^\sigma_n(x_0a_1x_1\dots a_nx_n)=
%        \prod_{0\le i<n}
%        E_2(\sigma(x_0a_1x_1\dots a_ix_i),a_{i+1}x_{i+1})
%\]
%\[
%    {\sf plays}^\sigma_n(\rho xb)={\sf plays}^\sigma_m(\rho x)\times E_2(\sigma(\rho x), b)
%\]

\subsection{Linear Functors}
\label{sect:linear_funcs}

We require a restriction of set-based polynomial 
functors (which are standard in coalgebra, see \cite[p. 49]{Jacobs_2016})
to \textit{linear functors}.
To reduce clutter, we often use juxtaposition to denote the product of functors
(we never need to denote the application of a constant functor),
e.g. $AB:=A\times B$.
Technically, linear functors are a class of functors built inductively
out of arbitrary coproducts and $1$, a constant functor
assigning every set to a singleton set.
Linear functors have a general form $H(Y)\cong A+BY$,
a consequence of $C(A+BY)\cong CA+CBY$
and $(A_0+B_0Y)+(A_1+B_1Y)\cong (A_0+A_1)+(B_0+B_1)Y$ (extended to arbitrary coproducts).
We reserve $H$ for a linear functor,
and use $H_X$ for the linear functor $H_X(Y)=X\times H(Y)$.

Shorthand $Y^n\cong Y^{\{0,\dots,n-1\}}\cong\prod_{0\le i\le n-1}Y$ is used
for lists of elements of $Y$ of length $n$.
Similarly, we use $Y^{<n}\cong \coprod_{0\le i < n}Y^i$ for
lists of elements of $Y$ of length less than $n$. % because it is the empty coproduct.
Finally, we use $Y^*=\coprod_{n\in\omega}Y^{<n}$ for finite lists.
We record the $n$-fold compositions and initial algebras
of $H$ and $H_X$ in the table below.

\begin{center}
\setlength{\tabcolsep}{12pt}
\begin{tabular}{ccc}
    $F$ & $F^n(Y)$ & Initial algebra
    \\
    \hline
    \\[-1em]
    $H$ & $A^{<n}B+A^nY$ & $A^*B$
    \\
    $H_X$ & $(XA)^{<n}XB+(XA)^nY$ & $(XA)^*XB$
\end{tabular}
\end{center}

\subsection{Distributive Laws}
\label{sect:distrib_laws}

We assume knowledge of the definition of a monad,
whose multiplication and unit, with functor part $T$,
are referred to as $\mu^T$ and $\eta^T$.
Let $P,\pown,Q$ denote the full powerset monad,
the non-empty powerset monad, and the non-empty finite powerset monad
respectively.
We have submonads $Q\rightarrowtail \pown \rightarrowtail P$.
Let $D$ denote the finite distribution monad,
mapping a set $X$ to the set of distributions over $X$ with finite support.
We will the monad morphism ${\sf supp}:T\to P$, for $T=Q,D$,
which is an inclusion when $T=Q$, and
maps a distribution to its {\it support} when $T=D$.

\begin{definition}
    Given two monads $S,T:\bb{C}\rightarrow\bb{C}$,
    a \emph{distributive law} of $T$ over $S$ 
    is a natural transformation $\delta:TS\rightarrow ST$ such
    that
    \begin{center}
    \setlength{\tabcolsep}{0.1em}
    \begin{tabular}{cccc}
    
    \tikz[overlay]{
        \node at (2.4,0) {\scalebox{0.8}{$[\mu^S]$} };
    }
    \begin{tikzcd}
        TSS \ar[r, "\delta"]
            \ar[d, "T\mu^S"]
        &[-1em] STS \ar[r, "S\delta"]
        &[-1em]  SST \ar[d, "\mu^S"]
        \\
        TS \ar[rr, "\delta"] & & ST
    \end{tikzcd}
    &
    \tikz[overlay]{
        \node at (2.4,0) {\scalebox{0.8}{$[\mu^T]$} };
    }
    \begin{tikzcd}
        TTS \ar[r, "T\delta"]
            \ar[d, "\mu^T"]
        &[-1em]  TST \ar[r, "\delta"]
        &[-1em]  STT \ar[d, "S\mu^T"]
        \\
        TS \ar[rr, "\delta"] & & ST
    \end{tikzcd}
    &
    \tikz[overlay]{
        \node at (1.9,0) {\scalebox{0.8}{$[\eta^S]$} };
    }
    \begin{tikzcd}
        &[-1em] T \ar[dl, "T\eta^S"'] \ar[dr, "\eta^S"] &[-1em] 
        \\
        TS \ar[rr, "\delta"] && ST
    \end{tikzcd}
    &
    \tikz[overlay]{
        \node at (1.9,0) {\scalebox{0.8}{$[\eta^T]$} };
    }
    \begin{tikzcd}
    &[-1em]  S \ar[dl, "\eta^T"'] \ar[dr, "S\eta^T"]
    \\
    TS \ar[rr,"\delta"] &&[-1em]  ST
    \end{tikzcd}
    \end{tabular}
    \end{center}
    commute.
    A \emph{weak distributive law} of $T$ over $S$ is a natural transformation $TS\rightarrow ST$
    such that $[\mu^S,\mu^T,\eta^S]$ hold.
    Note that any distributive law is a weak distributive law.
    A \emph{functor-monad distributive law} of a functor $F$ over a monad $S$
    is a natural transformation
    $FS\rightarrow SF$ such that $[\mu^S,\eta^S]$ (substituting $T$ for $F$)
    hold.
    In this paper, we refer to a functor-monad distributive law
    $FS\rightarrow SF$ as a distributive law
    when we are not considering any
    monad structure on $F$.
\end{definition}
We also require strength maps for a monad.
Given a monad $T:{\bf Set}\rightarrow{\bf Set}$,
the \textit{left strength map} is a functor-monad distributive law
${\sf stl}^T_{A,X}:A\times T(X)\rightarrow T(A\times X)$
between $A\times(-)$ and $T$ which interacts well
with the coherence isomorphisms for $\times$.
This map must exist and is unique in ${\bf Set}$.
There is also a unique {\it right strength map}
${\sf str}_{X,A}^T:T(X)\times A\rightarrow
T(X\times A)$, defined analogously.
%and a map ${\sf st}_{AX}:T(X^A)\rightarrow T(X)^A$
%defined by $u\mapsto T(\lambda f. f(a))(u)$.
A monad $T$ is \emph{commutative} when the following holds for all sets $X$ and $Y$:
\begin{eqnarray}
\label{eq:comm}    \mu^T_{X\times Y}\circ T({\sf str}_{X,Y})\circ {\sf stl}_{TX,Y} =\mu^T_{X\times Y}\circ T({\sf stl}_{X,Y})\circ {\sf str}_{X,TY}
\end{eqnarray}
%=\mu^T_{A\times X}\circ T({\sf stl}_{A,X})\circ {\sf str}_{A,TX}\]
%when $\mu^T_{A\times X}\circ T({\sf str}_{A,X})\circ {\sf stl}_{TA,X}
%=\mu^T_{A\times X}\circ T({\sf stl}_{A,X})\circ {\sf str}_{A,TX}$ holds
%for all sets $A$ and $X$.

\begin{example}
    \label{example-distributive-laws}
    This paper is concerned with the following distributive laws.
    \begin{enumerate}[label=(\roman*)]
        \item There is a weak distributive law (first appearing in \cite{garner_vietoris_2020})
        $\delta^{PP}:PP\rightarrow PP$:
        \[\delta^{PP}(\{U_i\}_{i\in I}):=\{\bigcup_{i\in I}V_i\mid
            \forall i \in I\,(\emptyset
            \subset V_i\subseteq U_i)\}\]
        that restricts to a law $\delta^{P\pown}:\pown P\rightarrow P\pown$.
        There is a variant $\delta^{PQ}:QP\rightarrow PQ$, which
        only takes finite subsets \cite{goy_compositionality_2021}:
        \[\delta^{PQ}(\{U_i\}_{i\in I}):=\{\bigcup_{i\in I}V_i\mid
            \forall i \in I\,(\emptyset
            \subset V_i\subseteq_\omega U_i)\}\]
        \item There is a weak distributive law $\delta^{PD}:DP\rightarrow PD$,
            distributing probability over nondeterminism \cite{goy_combining_2020}:
            \[\delta^{PD}([U_i\mapsto p_i]_{i\in I}):=
            \{\mu^D_X[\varphi_i\mapsto p_i]\mid\forall i \in I\, (\varphi_i\in D(X)
            \mt{ and }{\sf supp}\,\varphi_i\subseteq U_i) \}\]
        \item We have already noted that ${\sf str}$ is a functor-monad distributive law.
        The right strength is defined w.r.t.~the monoidal product
        $(\times,1)$ in ${\bf Set}$. 
        We can define another strength map w.r.t.~the monoidal product
        $(+,0)$ in ${\bf Set}$:
        \[ [\eta^T_{B+X}\circ{\sf inl}, T({\sf inr})]:B+T(X)\rightarrow
            T(B+X)\]
        Composing these functor-monad distributive laws
        we get a distributive law $\lambda_X:A+B(T(X))\rightarrow T(A+B(X))$,
        i.e. linear functors distribute over any ${\bf Set}$ monad.
    \end{enumerate}
\end{example}

\subsection{Composite Monads}
\label{sect:composite_monads}

Any weak distributive law induces a composite monad \cite{garner_vietoris_2020}.
To combine two ${\bf Set}$ monads $S$ and $T$ with $\delta^{ST} : TS \rightarrow ST$,
we define the associated \textit{convex closure operator} \cite{bonchi_convexity_2022}:
a natural transformation
${\sf cl}^{ST}_X:ST(X)\rightarrow ST(X)$,
defined by $S\mu^T\circ\delta^{ST}T\circ\eta^TST$.
The functor part of the composite monad is the image of the convex closure operator,
denoted by $\widetilde{ST}$.
The multiplication and unit of the composite,
relying on \cite[Lemma 2.10, p. 48]{goy_compositionality_2021},
can be given as the standard expressions
$\mu^S\mu^T\circ S\delta^{ST}$ and $\eta^S\eta^T$.
When $\delta^{ST}$ is a (full) distributive law,
$\widetilde{ST}\cong ST$.

\begin{example}
    The closure operator associated with $\delta^{PQ}$
    can be calculated:
    \[{\sf cl}_X^{PQ}(\mc{U})=\{\tts\bigcup\mc{V}\mid\emptyset\subset\mc{V}\subseteq_\omega\mc{U}\}\]
    It closes a set of finite subsets under finite, non-empty union.
    We call a set of finite subsets $\mc{U}$ \textit{convex}
    precisely when ${\sf cl}_X^{PQ}(\mc{U})=\mc{U}$,
    i.e. when $U,V\in\mc{U}\implies U\cup V\in \mc{U}$.
    The functor part $\widetilde{PQ}(X)$ is the set containing all convex sets of subsets.
    The closure operator for $\delta^{PD}$
    is given by
    \[{\sf cl}^{PD}_X(U)=\{\mu^D(\Phi)\mid \Phi\in DD(X),\, {\sf supp}\,\Phi\subseteq U\}\]
    Similarly, we call a set $U$ of distributions \textit{convex} precisely
    when ${\sf cl}^{PD}_X(U)=U$.
    $\widetilde{PD}(X)$ is the set of all convex sets of distributions.
    Both of these closure operators satisfy standard closure operator properties:
    $\mc{U}\subseteq{\sf cl}(\mc{U})$,
    $\mc{U}\subseteq\mc{U}'\Rightarrow {\sf cl}(\mc{U})\subseteq{\sf cl}(\mc{U}')$
    and ${\sf cl}(\mc{U})\subseteq {\sf cl}\circ{\sf cl}(\mc{U})$.
\end{example}

\begin{remark}
The category of algebras associated with the monad $\widetilde{PQ}$ is
a complete lattice $(X,\bigvee)$ with a meet semi-lattice (without top)
$(X,\wedge)$, such that
 $x \land \bigvee_i y_i = \bigvee_i (x\land y_i)$ holds,
 i.e. $\wedge$ distributes over $\bigvee$.
Similarly, the algebraic theory of the monad $\widetilde{PD}$
is a complete lattice $(X,\bigvee)$ and a convex algebra $(X,+_r)$
such that $+_r$ distributes over $\bigvee$.
This is discussed in \cite[Section 3.2.2]{goy_combining_2020}.
\end{remark}

\subsection{Kleisli Categories}

Given a monad $T:\bb{C}\rightarrow\bb{C}$, we can form its Kleisli category,
denoted ${\bf Kl}(T)$, which has the objects of $\bb{C}$
and homsets ${\bf Kl}(T)(X,Y):=\bb{C}(X,T(Y))$.
Kleisli composition $\odot$ is defined by composition with $\mu^T$:
$g\odot f:=\mu^T\circ T(g)\circ f$,
and has the identity $\eta^T_X$ for each object $X$.
We have an identity-on-objects functor $J_T:\bb{C}\rightarrow{\bf Kl}(T)$, defined as
$J_T(f)=\eta^T\circ f$, which is left adjoint to the forgetful functor ${\bf Kl}(T)\rightarrow\bb{C}$.
We denote $J_T(f)$ as $\widebar{f}$.
An \emph{extension} of $F:\bb{C}\rightarrow\bb{C}$ to ${\bf Kl}(T)$
is a functor $\widebar{F}:{\bf Kl}(T)\rightarrow {\bf Kl}(T)$ such
that $\widebar{F}J_T=J_TF$.
Kleisli extensions
are in one-to-one correspondence with
functor-monad distributive laws (see e.g. \cite[Chapter 5]{jacobs_coalgebraic_2008}).
Given a distributive law $\lambda:FT\rightarrow TF$,
we can define $\widebar{F}(f):=\lambda\circ F(f)$.
We often require extensions of iterated functors
in this paper,
so highlight that extension distributes over functor composition:
$\widebar{F^n}=\widebar{F}^n$.
We also have $\widebar{F}(\widebar{f})=\widebar{F(f)}$, by $\widebar{F}$ being an extension.

\subsection{Coalgebraic Traces}
\label{coalg-traces}
A \emph{finite trace} is an element of the initial $H$-algebra $A^*B$, i.e.~a finite sequence of observations $a_1\dots a_nb$ with $a_i \in A$ and $b\in B$.

To state the main trace semantics theorem from \cite{hasuo_generic_2007},
we require some domain-theoretic concepts.
An \emph{$\omega$-complete partial order} ($\omega$-cpo),
is a partial order $(X,\le)$ such that every $\omega$-chain $x_0\le x_1\le \cdots$ has a join
$\bigvee\{x_n\}_{n\in\omega}$.
A category is $\omega${\bf-cpo}-enriched when homsets can be equipped with $\omega$-cpo structure,
and composition preserves joins of $\omega$-chains separately in both arguments.
In an $\omega${\bf-cpo}-enriched category, a morphism $i:X\rightarrow Y$
is an \textit{embedding} precisely when there is a (necessarily unique) \textit{projection}
$p:Y\rightarrow X$ such that $p\circ i=\mt{id}_X$ and $i\circ p\le \mt{id}_Y$
(where $\le$ is the order on morphisms $X\rightarrow X$).
We also need the concept of a \textit{zero map} between objects $X$ to $Y$,
which exist in categories where there is a final and initial object $0$.
The zero map is the composite $X\xrightarrow{!}0\xrightarrow{!}Y$.
Finally, a functor $F:\bb{C}\rightarrow \bb{D}$ between
two $\omega${\bf-cpo}-enriched categories is \textit{locally monotone} when $f\le g$ implies $F(f)\le F(g)$.

We require two results: the first from work in domain theory \cite{smyth_category-theoretic_1982},
and the second which applies this result to obtain finite coalgebraic trace semantics \cite{hasuo_generic_2007}.
Both of these results feature in \cite[Chapter 5.3]{Jacobs_2016},
which our formulation is based on,
but phrased in terms of the slightly more general $\omega$-cpos (following \cite{hasuo_generic_2007})
rather than dcpos (the proofs in \cite{Jacobs_2016} only rely on joins of $\omega$-chains).

\begin{proposition}[{\normalfont \cite[Proposition 5.3.3]{Jacobs_2016}}]
    \label{prop:(co)limit-coincide}
    Let $\bb{C}$ be an $\omega${\bf-cpo}-enriched category
    with some $\omega$-chain $X_0\xrightarrow{i_0}X_1\xrightarrow{i_1}X_2\xrightarrow{i_2}\cdots$
    of embeddings,
    with colimit $(A,\kappa_n:X_n\rightarrow A)$ in $\bb{C}$.
    Each coprojection $\kappa_n$ is an embedding,
    call the associated projection $\pi_n:A\rightarrow X_n$.
    $(A,\pi_n)$ is a limit of the $\omega$-cochain of the projections
    $X_0\xleftarrow{p_0}X_1\xleftarrow{p_1}X_2\xleftarrow{p_2}\cdots$.
    The mediating morphism of a cone $(f_n:Y\rightarrow X_n)_{n\in\omega}$
    can be calculated with the join of an $\omega$-chain of morphisms $Y\rightarrow A$:
    \[\bigvee_{n\in\omega}(Y\xrightarrow{f_n}X_n\xrightarrow{\kappa_n}A)\]
\end{proposition}

The following theorem, from \cite{hasuo_generic_2007},
gives a categorical account of finite traces for systems modelled as coalgebras of certain functors on $\bf Set$.%is instantiated in ${\bf Set}$.

\begin{theorem}[{\normalfont \cite[Proposition 5.3.4]{Jacobs_2016}}]
    \label{theorem:traces}
    Given a monad $T:{\bf Set}\rightarrow{\bf Set}$ and a functor $F:{\bf Set}\rightarrow{\bf Set}$,
    if
    \begin{enumerate}
        \item ${\bf Kl}(T)$ is $\omega${\bf-cpo}-enriched
        \item $T0 \to 1$ is an isomorphism
        \item Zero maps are bottom elements in the Kleisli homsets
        \item We have a distributive law $\lambda:FT\rightarrow TF$
        \item The extension $\widebar{F}:{\bf Kl}(T)\rightarrow{\bf Kl}(T)$ is locally monotone
        \item The initial $F$-algebra $F(I_F)\xrightarrow[\sim]{\alpha}I_F$ is the colimit of $0\xrightarrow{!} F(0)\xrightarrow{F!} F^2(0)\rightarrow\cdots$
    \end{enumerate}
    Then $J$ takes the
    initial $F$-algebra to a final $\widebar{F}$-coalgebra, i.e.
    $J(\alpha^{-1}):I_F\rightarrow\widebar{F}(I_F)$ is final in ${\bf Coalg}_{\bf Kl{(T)}}(\widebar{F})$.
\end{theorem}
In the above theorem, the monad $T$ describes the type of computation (e.g.~non-deterministic),
whereas the functor $F$ specifies what can be observed along single computation paths.
The final $\widebar{F}$-coalgebra gives us a unique $\widebar{F}$-coalgebra
morphism
from any $\widebar{F}$-coalgebra $c:X\to TF(X)$ into $J(\alpha^{-1}):I_F\to TF(I_F)$.
This coalgebra morphism $(X,c)\to (I_F,J(\alpha^{-1}))$
is usually referred to as the \textit{trace map},
because it recovers the definition of finite traces
in many different examples \cite{hasuo_generic_2007}.

\section{Traces and Executions}
\label{sect:traces_and_execs}

We now show that a two-player game modelled as a coalgebra
$X\to \widetilde{PT}H(X)$, is susceptible to
the coalgebraic framework for finite trace semantics given
by Hasuo et al. \cite{hasuo_generic_2007}.
We have two different choices for the monad $T$.
We either take $T=Q$,
which gives us something similar
to the two-player games discussed in Section~\ref{sect:two_player_games},
or $T=D$, which give the MDPs
from Section~\ref{sect:mdps}.

Thus, we must show that
the requirements of Theorem~\ref{theorem:traces}
are met by the monads $\widetilde{PQ}$ and $\widetilde{PD}$,
and a linear functor $H$.
To motivate our choice of monads $P$ and $Q$,
and our restriction to linear functors, we also
examine what doesn't work.
Two of these failures exist, in some form, as mistakes in the literature.
These are outlined in Counterexamples \ref{counter:left_strict} and \ref{counter:commutativity}.
Once we have established
the existence of trace maps, we introduce \emph{execution maps} (as a special type of trace maps) and show that every trace map factors through an execution map.

\subsection{Kleisli Enrichment and Zero Maps}

This subsection examines the requirements for Theorem~\ref{theorem:traces}
for the monad $\widetilde{PT}$,
namely whether $(i)$ ${\bf Kl}(\widetilde{PT})$ is $\omega$-{\bf cpo} enriched,
$(ii)$ $\widetilde{PT}(0)\cong 1$ and $(iii)$ the zero map
$X\rightarrow 0\rightarrow Y$ in ${\bf Kl}(\widetilde{PT})$ is the least element
in ${\bf Kl}(\widetilde{PT})(X,Y)$.
When choosing a powerset monad $T$ to model the environment non-deterministically,
we explain why the finiteness and non-emptiness restrictions are required.
To ensure $(ii)$, we must choose a non-empty variant.
Counterexample~\ref{counter:omega-cpo} exhibits that ${\bf Kl}(\widetilde{P\pown})$
is in fact not $\omega$-{\bf cpo} enriched (at least not in the natural way), meanwhile
Proposition~\ref{prop:kl_cpo_enriched} shows that ${\bf Kl}(\widetilde{PQ})$ is.
Hence, we are left with the choice of finite non-empty powerset $Q$ to model
a non-deterministic environment.
For a probabilistic environment, the standard choice of the finite distribution
monad $D$ works out of the box: it already contains finite collections,
and distributions summing to one mean they can't be ``empty'' (have no non-zero elements).

To show ${\bf Kl}(\widetilde{PT})$ is $\omega${\bf -cpo} enriched,
we follow the general approach
in \cite[Chapter 5]{bonchi_convexity_2022}.
Let $T=P,Q$ or $\pown$ (we briefly consider $\pown$ only to rule it out after
the next counterexample).
Firstly, $\widetilde{PT}(Y)$ is a complete lattice with
an order given by standard subset inclusion and joins given by
$\bigvee:={\sf cl}^{PT}\circ\bigcup$.
This order is lifted pointwise to an order $\sqsubseteq$
on morphisms $X\rightarrow \widetilde{PT}(Y)$.
What is left is to show that Kleisli composition preserves
joins in each argument:
\[ 
    \bigvee_{n\in\omega}g_n\odot f=\bigvee_{n\in\omega}(g_n\odot f)
    \hspace{3em}
    g\odot\bigvee_{n\in\omega}f_n=\bigvee_{n\in\omega}(g\odot f_n)
\]
We see that the left condition breaks in the case of $\widetilde{P\pown}$.

\begin{counterexample}
    \label{counter:omega-cpo}
    Take $f:\{x\}\rightarrow\widetilde{P\pown}(\omega)$ and a $\omega$-chain $\{g_i:\omega\rightarrow\widetilde{P\pown}(\{w,z\})\}_{i\in\omega}$.
    Define $f(x):=\{\omega\}$ and
    \[g_i(n)=\begin{cases}
        \{\{a\}\} & \mt{if }n\le i
        \\
        \{\{a\},\{b\},\{a,b\}\} & \mt{otherwise}
    \end{cases}\]
    We find $\{b\}\in \bigvee\{g_i\}_{i\in\omega}\odot f(x)$ because $\{b\}\in g_i(i)$ for each $i\in\omega$.
    However $\{b\}\not\in\bigvee\{g_i\odot f\}_{i\in\omega}(x)$
    because there is no $i\in\omega$ such that $\{b\}\in g_i(n)$ for all $n\in\omega$.
    Hence ${\bf Kl}(\widetilde{P\pown})$ is not {\bf $\omega$-cpo}-enriched.
    The same reasoning will apply to the monotone neighbourhood monad.
\end{counterexample}

Arguably the most difficult property to establish
is that ${\bf Kl}(\widetilde{PT})$ is $\omega${\bf-cpo}-enriched.
Note that enrichment in the
category of directed complete partial orders
has been proven for the monad considered in \cite{jacobs_coalgebraic_2008},
in \cite{brengos} and \cite{goncharov}.

\begin{proposition}
    \label{prop:kl_cpo_enriched}
    ${\bf Kl}(\widetilde{PT})$ is $\omega${\bf-cpo}-enriched, for $T=Q,D$.
\end{proposition}
\begin{proof} 
    \hyperref[appendix:kl_cpo_enriched]{Appendix}.
\end{proof}

%\vspace{-1.5em}

%\paragraph{Left-strictness}
Bonchi and Santamaria \cite[Theorem 5.14]{bonchi_convexity_2022}
prove that composition in the Kleisli category ${\bf Kl}(\widetilde{PT})$ is \textit{left strict}: meaning we have that
$\bot\odot f=\bot$.
We find this to be a mistake, because left strictness is a sufficient
condition for the
isomorphism $\widetilde{PT}(0)\cong 1$ (see \cite[Lemma 3.5]{hasuo_generic_2007}),
which also does not hold for the choice of $T$ in \cite{bonchi_convexity_2022}.
We construct a direct counterexample for $T=P_f$,
the finite powerset monad, below.

\begin{counterexample}
    \label{counter:left_strict}
    Take $f:\{x\}\rightarrow \widetilde{PP_f}(Y)$
    as $x\mapsto \{\emptyset\}$.
    We have $\bot\odot f(x)=\{\emptyset\}$,
    which means this composite is not equal to $\bot$
    (we should have $x\mapsto\emptyset$), further
    details are in the \hyperref[appendix:left_strict_counter_info]{Appendix}.
    \[\bot\odot f(x)=\mu^P\mu^{P_f}\circ P\delta^{P_fP}\circ PP_f(\bot)(\{\emptyset\})
    =\mu^P\mu^{P_f}\circ P\delta^{P_fP}(\{\emptyset\})
    =\mu^P\mu^{P_f}(\{\{\emptyset\}\})
    =\{\emptyset\}
    \]
\end{counterexample}

The condition of left strictness says that \textit{controller deadlocks are preserved by pre-composition}.
In our counterexample, the environment deadlocks first (in the $f$ above),
and as a result, the environment also deadlocks in the composition:
$\bot\odot f(x) =\{\emptyset\}$.
This is different from the controller deadlocking: $\bot(x) =\emptyset$. 
This counterexample also exists in monad considered in \cite{bonchi_convexity_2022}
(we believe for any positive semifield:
by having $f$ point to the singleton ``null distribution'').

The requirement which is violated in Theorem~\ref{theorem:traces}
is that $\widetilde{PP_f}(0)\cong\{\emptyset,\{\emptyset\}\}\not\cong 1$.
A natural fix is to prevent one player from deadlocking;
we choose to disallow environment deadlocks, as this ensures bottom elements remain intact.

\begin{proposition}
    Zero maps in ${\bf Kl}(\widetilde{PT})$ form bottom elements,
    for $T=Q,D$.
\end{proposition}
\begin{proof}
    We have $\widetilde{PT}(0)=\{\emptyset\}$ for $T=Q,D$, hence:
    \begin{center}
    \begin{tikzcd}
        X \ar[r, "!"]
        & \widetilde{PT}(0) \ar[r, "\widetilde{PT}(!_Y)"]
        &[2em] \widetilde{PT}\widetilde{PT}(Y) \ar[r, "\mu^{\widetilde{PT}}"]
        & \widetilde{PT}(Y)
        \\[-1.8em]
        x \ar[r, mapsto]
        & \emptyset \ar[r, mapsto]
        & \emptyset \ar[r, mapsto]
        & \emptyset
    \end{tikzcd}
    \end{center}
    because the multiplication $\mu^{\widetilde{PT}}$ (and the distributive law $\delta^{PT}$)
    preserves the empty set.
\end{proof}

\subsection{Commutativity and a Functor-Monad Distributive Law}

Now we have an appropriate $\omega${\bf-cpo} structure on the Kleisli homsets,
we turn our attention to the behaviour functor $F:{\bf Set}\rightarrow{\bf Set}$,
and meeting the requirements $(iv),(v)$ and $(vi)$ of
Theorem~\ref{theorem:traces}.
For $(iv)$, we require a functor-monad distributive law
$F\widetilde{PT}\rightarrow \widetilde{PT}F$.
It is known that any polynomial functor distributes over a
commutative monad \cite[Lemma 2.4]{hasuo_generic_2007},
and \cite[Lemma 5.2]{jacobs_coalgebraic_2008} claims their
similar monad is commutative.
Unfortunately, this appears to be a mistake,
as the following counterexample demonstrates.
Note that this was also noticed recently
(and independently) in \cite{staton}.
%Recall the diagram that must commute for a monad $T$ to be commutative:
%\begin{center}
%    \begin{tikzcd}
%    &[2em] T(T(X)\times Y) \ar[r, "T(\mt{str}^T)"] &[2em] TT(X\times Y) \ar[dr, "\mu^T"]
%    \\[-2em]
%        T(X)\times T(Y) \ar[ur, "\mt{stl}^T"] \ar[dr, "\mt{str}^T"'] &&&[2em] T(X\times Y)
%    \\[-2em]
%    & T(X\times T(Y)) \ar[r, "T(\mt{stl}^T)"'] & TT(X\times Y) \ar[ur, "\mu^T"']
%    \end{tikzcd}
%\end{center}

\begin{counterexample}
    \label{counter:commutativity}
    %The commuting diagram required to prove commutativity of $\widetilde{PQ}$ fails
    Condition (\ref{eq:comm}), required for the commutativity of $\widetilde{PQ}$, fails
    for the sets $\{\{x_1\},\{x_2\},\{x_1,x_2\}\}$ and $\{\{y_1,y_2\}\}$.
    If we think of these sets as one-step interactions in games,
    composing them in different orders
    give different sets that can be forced.
    This counterexample also exists for $\widetilde{PD}$,
    the monotone neighbourhood monad, and the semifield monads considered in \cite{jacobs_coalgebraic_2008,bonchi_convexity_2022}.
    Details are in the \hyperref[appendix:commutativity_counter_info]{Appendix}.
\end{counterexample}

The underlying reason for this failure
is that $\delta^{PQ}$ is not a \textit{weak distributive law of commutative monads}
in the sense of \cite{jacobs_semantics_1994}.
To solve this issue, we restrict $F$ to be a linear functor $H$,
and obtain the required functor-monad distributive law
$\lambda:H\widetilde{PT}\rightarrow \widetilde{PT}H$
automatically (recall Example \ref{example-distributive-laws} (ii)).
Point $(iv)$ automatically holds for a linear functor (see Section \ref{sect:distrib_laws}),
and so does point $(vi)$.
Finally, $(v)$ is easily proven.

\begin{proposition}
    \label{prop:locally_monotone}
    The extension $\widebar{H}:{\bf Kl}(\widetilde{PT})
        \rightarrow{\bf Kl}(\widetilde{PT})$ is
    locally monotone.
\end{proposition}
\begin{proof}
    Given some $f,g:X\rightarrow\widetilde{PT}(Y)$ such that
    $f\sqsubseteq g$, we must establish $\lambda\circ Hf\sqsubseteq \lambda\circ Hg$,
    where $\lambda:H\widetilde{PT}\rightarrow \widetilde{PT}H$ is the
    functor-monad distributive law, this is easily verified.
\end{proof}

\subsection{Trace Maps}

\begin{wrapfigure}{r}{0cm}
\hspace{-2em}
\begin{tikzcd}
    X
    \ar[r, "{\sf tr}_c", dashed]
    \ar[d, "c"]
    &[2em]
    A^*B
    \\[-0.5em]
    \widebar{H}(X)
    \ar[r, "\widebar{H}({\sf tr}_c)", dashed]
    &
    \widebar{H}(A^*B)
    \ar[u, "\rotatebox{90}{$\sim$}"']
\end{tikzcd}
\end{wrapfigure}
Instantiating Theorem~\ref{theorem:traces} for $\widetilde{PT}H$-coalgebras gives
us a final $\widebar{H}$-coalgebra in ${\bf Kl}(\widetilde{PT})$.
Hence, given a coalgebra $c:X\rightarrow\widebar{H}(X)$,
we can form the \textit{trace map} ${\sf tr}_c:X\rightarrow A^*B$
in ${\bf Kl}(\widetilde{PT})$
as the unique $\widebar{H}$-morphism into the final $\widebar{H}$-coalgebra.
Recall (from Section~\ref{sect:linear_funcs}) that $A^*B$ is the carrier of the initial $H$-algebra.

This takes a state and maps it to a convex set
of collections of $A^*B$.
We can give this map directly using
Proposition \ref{prop:(co)limit-coincide}.
Let $\kappa_n:H^n(0)\rightarrow A^*B$ be the coprojection in ${\bf Set}$
over the initial sequence, and recall that $c_n$
is the iterated coalgebra map
$X\xrightarrow{c}\widebar{H}(X)\xrightarrow{\widebar{H}(c)}\dots\xrightarrow{\widebar{H^{n-1}}(c)} \widebar{H^n}(X)$. Then
\[ {\sf tr}_c=\bigvee_{n\in\omega}(X\xrightarrow{c_n}\widebar{H^n}(X)\xrightarrow{\widebar{H^n}(!)}\widebar{H^n}(0)\xrightarrow{\widebar{\kappa_n}}A^*B)\,. \]

\subsection{Execution Maps}

Recall that a finite trace is an element of $A^*B$. Traces describe the \emph{observable} outcomes of plays.
To capture the actual plays, we also need to incorporate information about the controller states they visit.
We call the resulting concept executions. Formally, an \emph{execution} is an element of the initial algebra $(XA)^*XB$ of the functor $H_X$ from Section~\ref{sect:linear_funcs}, i.e.~a sequence $x_0a_1x_1\dots a_nx_nb$.

%\begin{remark}
%\label{rem:plays-execs}
%    We observe that executions are very similar to completed plays in the standard sense. The only difference is that the former do not record environment states. We justify this slight mismatch by observing that, since the objective of a game is to force certain ("good") sets of traces, recording actual (environment) states in the play history is unnecessary. The same applies to partial executions and partial plays.
%\end{remark} 
\begin{wrapfigure}{r}{0cm}
\hspace{-2em}
\begin{tikzcd}
    X
    \ar[r, "{\sf exec}_c", dashed]
    \ar[d, "c^*"]
    &[2em]
    (XA)^*XB
    \\[-0.5em]
    \widebar{H_X}(X)
    \ar[r, "\widebar{H_X}({\sf exec}_c)", dashed]
    &
    \widebar{H_X}((XA)^*XB)
    \ar[u, "\rotatebox{90}{$\sim$}"']
\end{tikzcd}
\end{wrapfigure}
Following \cite{cirstea_maximal_2011},
we modify an $\widebar{H}$-coalgebra $c$ to include state information,
obtaining a $\widebar{H_X}$-coalgebra $c^*$:
\[c^*:=(X\xrightarrow{\langle\mt{id},c\rangle} X\times TH(X)\xrightarrow{\mt{stl}}T(X\times H(X))\]
Then, as $H_X$ is still a linear functor,
we can apply Theorem \ref{theorem:traces} again
to obtain the \textit{execution map}
${\sf exec}_c={\sf tr}_{c^*}:X\rightarrow (XA)^*XB$ in ${\bf Kl}(\widetilde{PT})$.

We will use the following standard result, from the theory of coalgebras,
to prove
that a trace map for a coalgebra $c$ arises from 
an execution map for $c^*$, followed by forgetting
the states. %Although we could not find a reference for it,
%it is a standard result in the theory of coalgebras.

\begin{proposition}
    \label{prop:beh_factors_nat_trans}
    Suppose functors $F,G:\bb{C}\rightarrow\bb{C}$
    with final coalgebras $\zeta_F:Z_F\xrightarrow{\sim} F(Z_F)$
    and $\zeta_G:Z_G\xrightarrow{\sim} G(Z_G)$ respectively.
    A natural transformation $\alpha:F\rightarrow G$ induces a $G$-coalgebra morphism
    $f_\alpha:Z_F\rightarrow Z_G$
    and a functor $E_\alpha:{\bf Coalg}(F)\rightarrow{\bf Coalg}(G)$
    such that for any coalgebra $c:X\rightarrow F(X)$
    we have that
    \begin{center}
        \begin{tikzcd}
            X \ar[r, dashed, "\mt{beh}_c"] \ar[rr, bend right, dashed, "\mt{beh}_{E_\alpha(c)}"']&
                Z_F \ar[r, "f_\alpha", dashed] & Z_G
        \end{tikzcd}
    \end{center}
\end{proposition}
\begin{proof}\hyperref[appendix:details_beh_factors_nat_trans]{Appendix}.
\end{proof}

\begin{proposition}[Traces via Executions]
    \label{prop:traces_via_executions}
    The projection $\pi_2:H_X\rightarrow H$ is a natural transformation,
    inducing a function $f_{\pi_2}:(XA)^*XB\rightarrow A^*B$.
    The trace map can be factored via the execution map and $f_{\pi_2}$:
    we have $f_{\widebar{\pi_2}}\odot{\sf exec}_c={\sf tr}_c$ for any $c:X\rightarrow H(X)$.
\end{proposition}
\begin{proof}
    This follows from Proposition \ref{prop:beh_factors_nat_trans}
    with the fact that $E_{\widebar{\pi_2}}(c^*)=\widebar{\pi_2}\odot c^*=\mu\circ\eta\circ\pi_2\circ c^*=c$.
\end{proof}

\section{Strategies}
\label{sect:strategies}

This section shows how we recover (i) the usual notion of strategy in a two-player
game, as a chain of maps
in the Kleisli category of the monad $T$, and (ii) the outcome of a strategy,
by lifting this chain to ${\bf Kl}(\widetilde{PT})$
(where our games, and the appropriate limit, exist).
Recall that a strategy $\sigma$
maps incomplete plays which conform to $\sigma$
ending in a controller state $x$,
to a successor of $x$ (an environment state). The latter is an element of $TH(X)$,
and thus the natural home for strategies is ${\bf Kl}(T)$.

The reason to lift each $\sigma_n:\mt{Im}
(\sigma_n)\to\widehat{H^n_X}(X)$ to ${\bf Kl}(\widetilde{PT})$
is threefold.
Firstly, we need $(XA)^*XB$ to be the limit of
the final sequence, which gives
us a unique mediating map $1\to (XA)^*XB$ in ${\bf Kl}(\widetilde{PT})$.
We also need to be able to reason that $\sigma_n$ picks
a successor in $c$, which is a morphism in ${\bf Kl}(\widetilde{PT})$,
the order structure in ${\bf Kl}(\widetilde{PT})$ lets us do this.
The final reason is conceptual,
we want the outcome of a strategy to be from the \textit{controllers perspective},
which is captured in ${\bf Kl}(\widetilde{PT})$, rather
than ${\bf Kl}(T)$ which is the environments perspective.
This distinction is particularly prominent when a particular strategy
does not force a collection of completed plays (see Remark~\ref{finite-approximation-remark}).

This lifting is provided by a functor $K$ in the following proposition.

%(the other variant $K_T$ will be used in Section~\ref{sect:goals}).

\begin{proposition}
    \label{prop:kl_functors}
    There is an identity-on-objects functor
    $K:{\bf Kl}(T)\rightarrow{\bf Kl}(\widetilde{PT})$,
    %and $K_T:{\bf Kl}(S)\rightarrow {\bf Kl}(\widetilde{ST})$,
    defined by
    \[K(X\xrightarrow{f}T(Y)):=(X\xrightarrow{f}T(Y)\xrightarrow{\eta^P}PT(Y)
    \xrightarrow{{\sf cl}^{PT}}\widetilde{PT}(Y))\]
    %\[ K_T(X\xrightarrow{f}SY):=(X\xrightarrow{f}SY\xrightarrow{S\eta^T}STY
    %    \xrightarrow{{\sf cl}^{ST}}\widetilde{ST}(Y))\]
    %which make the following diagram commute:
    %\begin{center}
    %    \begin{tikzcd}
    %        &
    %        {\bf Kl}(T)
    %            \ar[dr, "K_S"]
    %        \\
    %        \bb{C} \ar[rr, "J_{\widetilde{ST}}"]
    %            \ar[ur, "J_T"]
    %            \ar[dr, "J_S"']
    %        &
    %        &
    %        {\bf Kl}(\widetilde{ST})
    %        \\
    %        &
    %        {\bf Kl}(S)
    %            \ar[ur, "K_T"']
    %    \end{tikzcd}
    %\end{center}
    %(note that here formally ${\sf cl}^{ST}$ is the natural transformation
    %named $\pi F^T$ in \cite{garner_vietoris_2020,goy_compositionality_2021}).
\end{proposition}
\begin{proof}
    Functoriality of $K$ follows from ${\sf cl}^{PT}\circ\eta^P$
    %and ${\sf cl}\circ S\eta^T$ being
    being a monad morphism,
    see \cite[Proposition 2.7]{goy_compositionality_2021}.
    %That the diagram commutes follows from
    %\cite[Lemma 2.10]{goy_compositionality_2021}.
\end{proof}

We are now ready to give a categorical definition of a strategy.
It will involve a chain of maps that live in ${\bf Kl}(T)$,
so we also need liftings and extensions to ${\bf Kl}(T)$.
To differentiate Kleisli categories, we
use $\widehat{f}$ for lifting functions to ${\bf Kl}(T)$,
and $\widehat{F}$ for extensions to ${\bf Kl}(T)$.

In the following definition,
we also use the projection $\pi_1:X(B+AX)\rightarrow X$ out of the product.
For the reader's intuition,
we spell out the map
\[ H^n_X(\pi_1): (XA)^{<n}XB+(XA)^nX(B+AX)\rightarrow (XA)^{<n}XB+(XA)^nX \]
as preserving complete plays of length $n$ or less,
and mapping complete and partial plays of length $n+1$
to the partial plays of length $n$ which they extend.
We define $\mt{Im}(-)$ as
mapping a ${\bf Kl}(T)$ morphism 
$f:X\rightarrow TY$ to
the set $\bigcup\circ P({\sf supp})\circ Pf(X)$.

\begin{definition}
    \label{def:strategy}
    Let $(X,c)$ be an $\widebar{H}$-coalgebra,
    and $x\in X$ be a state.
    A {\it strategy $\sigma$ from $x$}
    is a chain of maps $\{\sigma_n\}_{n\in\omega}$ in ${\bf Kl}(T)$,
    where $\sigma_0:1\rightarrow X$ picks $\eta^T(x)$, while 
    $\sigma_{n+1}:\mt{Im}(\sigma_n)\rightarrow \widehat{H^{n+1}_X}(X)$
    for $n\in\omega$, is such that
    \begin{center}
    \begin{tikzcd}
        \widehat{H^n_X}(X) &[2em]  \widehat{H^{n+1}_X}(X) \ar[l, "\widehat{H^n_X}(\widehat{\pi_1})"']
        \\[-0.5em]
        \mt{Im}(\sigma_n) \ar[u, >->] \ar[ur, "\sigma_{n+1}"']
    \end{tikzcd}
    \hspace{3em}
    \tikz[overlay]{
        \node at (1.75,0.1) {{\scalebox{1}{$\sqsupseteq$}}};
    }
    \begin{tikzcd}
        \widebar{H^n_X}(X) \ar[r, "\widebar{H^n_X}(c^*)"]   &[2em] \widebar{H^{n+1}_X}(X)
        \\[-0.5em]
        \mt{Im}(\sigma_n) \ar[ur, "K(\sigma_{n+1})"'] \ar[u, >->]
    \end{tikzcd}
    \end{center}
    commute. Denote the set of strategies from $x$ as $\Sigma_c(x)$.
    Define an $n${\it-step strategy from $x$}
    as a finite collection of maps $\{\sigma_0,\dots, \sigma_n\}$
    with the same conditions.
    Denote the set of $n$-step strategies from $x$ by $\Sigma_{n,c}(x)$.
\end{definition}

The left condition states that $\sigma_{n+1}$ preserves completed plays
and extends incomplete plays (by either completing them or by adding
a successor). The right condition states
that $\sigma_n$ chooses out of the successors in $(X,c)$.
Notice that
$\sigma_n$ cannot force a play ending in a state $x\in X$
with no successors (i.e. a controller deadlock),
because then $\sigma_{n+1}$ is not definable,
as there is no way to satisfy the right condition above. %of being a strategy.
Hence, there are $n$-step strategies $\{\sigma_0,\dots,\sigma_n\}$ which cannot
be extended to a (full) strategy $\{\sigma_n\}_{n\in\omega}$.

We recall (see Section~\ref{sect:two_player_games}
and \ref{sect:mdps}), that a standard strategy $(x,\sigma)$
for a $\widetilde{PT}H$-coalgebra
is a partial function from $(XA)^*X$ to controller states, defined over
exactly the partial plays which conform to it,
which extends partial plays ending in $x$ with a move chosen from the
controllers edge relation.
With a few inconsequential assumptions
on the two-player games and MDPs defined in Section \ref{sect:two_player_games} and \ref{sect:mdps},
they are equivalently coalgebras $X\to \widetilde{PT}H(X)$,
for $T=Q$ and $T=D$ respectively.
Furthermore, standard strategies are then in one-to-one correspondence
with strategies defined in Definition \ref{def:strategy}.

%\begin{proposition}
%    Fix a $\widetilde{PT}H$-coalgebra $(X,c)$ and a state $x$.
%    There is a bijective correspondence between strategies at $x$
%    and pointed standard strategies starting in $x$.
%\end{proposition}
%\begin{proof}
%    Given a strategy $\sigma$ at a state $x$, we can construct a pointed standard strategy $(x,\tau)$
%    with $\tau(\rho)=\sigma_n(\rho)$ for each $\rho\in\mt{Im}(\sigma_n)$ with $\rho\in (XA)^nX$.
%    Conversely, given a pointed standard strategy $(x,\tau)$, we can define $\sigma$ inductively:
%    \[ \sigma_0(*):=\eta^T(x)\hspace{4em}\sigma_{n+1}(\rho):=\begin{cases}
%            \tau(\rho) & \mt{if }\rho\in (XA)^nX
%            \\
%            \rho & \mt{if }\rho\in (XA)^{<n}XB
%        \end{cases}
%    \]
%    Checking the constraints and that these constructions are mutual inverses is easy.
%\end{proof}

Given an strategy $\sigma$ at a state $x$, the set of plays from $x$
up to some depth $n$ naturally arises from composition in ${\bf Kl}(T)$.
For example, in the $T=Q$ case,
composition takes a union at each step, so composing with $\sigma_{n+1}$ takes a set
of complete plays (of length $\le n$) and partial plays (of length $n$)
to a set of complete (of length $\le n+1$) and partial plays (of length $n+1$).
This is the content of our next definition. Again,
this can be seen to agree with the $n$-step partial outcome
of a standard strategy in discussed Section \ref{sect:prelim}.

We require the dashed morphism and injection in the commutative
diagram below
(the unique surjective-injective factorisation of $\sigma_{n+1}$ in ${\bf Set}$),
which can be lifted into ${\bf Kl}(T)$ with $J_T$.
\begin{center}
    \begin{tikzcd}
        & TH_X^{n+1}(X)
        \\[-0.5em]
        \mt{Im}(\sigma_n) \ar[ur, "\sigma_{n+1}"] \ar[r, dashed] & \mt{Im}(\sigma_{n+1}) \ar[u, >->]
    \end{tikzcd}
\end{center}

\begin{definition} Let $\sigma$ be a strategy from $x$.
    For $n\in \omega$, define ${\sf plays}_n^\sigma(x):1\rightarrow \widehat{H^n_X}(X)$ as the composition of the dashed arrows to $\mt{Im}(\sigma_{n})$ with the injection
    to $\widebar{H^n_X}(X)$.
\begin{center}
\begin{tikzcd}
     & X & \widehat{H_X}(X) & \widehat{H^2_X}(X) & 
    \\[-0.5em]
    1 \ar[ur, "\sigma_0"] \ar[r, dashed]
    & \mt{Im}(\sigma_0) \ar[u, >->] \ar[r, dashed] \ar[ur, "\sigma_1"]& \ar[r, dashed] \mt{Im}(\sigma_1) \ar[u,>->] \ar[ur, "\sigma_2"] &
    \mt{Im}(\sigma_2) \ar[u,>->] \ar[r, dashed] & \cdots
\end{tikzcd}
\end{center}
    We refer to ${\sf plays}_n^\sigma(x)$ as the $n$-step partial outcome of $\sigma$.
\end{definition}

For simplicity, the following discusses $T=Q$, but the same
reasoning applies to $T=D$ too.
The map ${\sf plays}_n^\sigma(x):1\rightarrow \widehat{H^n_X}(X)$ gives the set
of complete plays of length less than or equal to $n$ and partial plays of length $n$,
which can be obtained from fixing the controller strategy $\sigma$ from $x$.
We now want the \textit{completed outcome} of a strategy:
all $\sigma$-conform complete plays from a state $x$.
It is tempting to try to apply Theorem \ref{theorem:traces} 
to obtain a final $\widebar{H}$-coalgebra in ${\bf Kl}(Q)$,
and then to obtain a mediating map from a cone over the final sequence.
However we can see, due to $Q$ being \textit{finite} subsets,
that joins of chains of morphisms are not guaranteed to exist.
For example, the union of $\{x_0\}\subseteq\{x_0,x_1\}\subseteq \{x_0,x_1,x_2\}\subseteq \dots$
is not a finite set.
This failure means ${\bf Kl}(Q)$ is not $\omega${\bf-cpo} enriched, so
Theorem \ref{theorem:traces} and Proposition \ref{prop:(co)limit-coincide} do not apply.
%Thinking of $U_n$ in the sequence $U_0\subseteq U_1\subseteq\dots$
%as containing the completed plays of length less than $n$ resulting from a strategy $\sigma$,
%a defect of using $Q$ is that it cannot talk about
%the set of all completed plays coming from $\sigma$, as there may be infinitely many of them
%(an example of this is given in Remark \ref{finite-approximation-remark}).
\

We solve this problem
by using $K$ to lift the chain to ${\bf Kl}(\widetilde{PT})$,
where we use it to form a cone over the final sequence.
Note that, from the proof of Theorem \ref{theorem:traces} in \cite{Jacobs_2016},
we know that
$0\xleftarrow{!}\widebar{H_X}(0)\xleftarrow{\widebar{H_X}(!)}
\widebar{H^2_X}(0)\leftarrow\cdots$
is the final sequence of $\widebar{H_X}$ in ${\bf Kl}(\widetilde{PT})$.

\begin{proposition}
\label{prop:cone}
    Fix a strategy $\sigma$ from $x$, in some $\widebar{H}$-coalgebra.
    $(1\xrightarrow{K({\sf plays}^\sigma_n(x))}\widebar{H^n_X}(X)\xrightarrow{\widebar{H^n_X}(!)}\widebar{H^n_X}(0))_{n \in \omega}$
    defines a cone %\marginpar{not the projections, right? \color{blue} They form the maps of the cone, are these called projections? I use projections for the limiting cone...} 
    over the final sequence of $\widebar{H_X}$
    $0\xleftarrow{!}\widebar{H_X}(0)\xleftarrow{\widebar{H_X}(!)}\widebar{H^2_X}(0)\leftarrow\cdots$.
\end{proposition}
\begin{proof}
    The top square commutes by finality (0 is final in ${\bf Kl}(\widetilde{PT})$),
    the rest commute from definitions.
    \begin{center}
    \begin{tikzcd}
        &[3em]
        \widebar{H^n_X}(0) &[3em] \ar[l,"\widebar{H^n_X}(!)"'] \widebar{H^{n+1}_X}(0)
        \\[-1em]
        &
        \widebar{H^n_X}(X) \ar[u, "\widebar{H^n_X}(!)"] &[3em] \ar[l,"\widebar{H^n_X}(\widehat{\pi_1})"'] \widebar{H^{n+1}_X}(X) \ar[u, "\widebar{H^{n+1}_X}(!)"]
        \\[-1em]
        1 \ar[r,dashed] \ar[ur, "K({\sf plays}^\sigma_n(x))"] \ar[rru, bend right, "K({\sf plays}^\sigma_{n+1}(x))"']
        & \mt{Im}(\sigma_n) \ar[u, >->] \ar[ur,"K(\sigma_{n+1})"{yshift=-3px}]
    \end{tikzcd}
    \end{center}
\end{proof}

We are now ready to define the completed plays that conform to an $(X,c)$-strategy from a state $x$.

\begin{definition}
    \label{def:strategy_outcome}
    Fix an $\widebar{H}$-coalgebra $(X,c)$ and a state $x\in X$.
    For $\sigma\in\Sigma_c(x)$, define
    ${\sf plays}_{c}^{\sigma}(x):1\rightarrow (XA)^*XB$ 
    in ${\bf Kl}(\widetilde{PT})$,
    as the unique mediating map
    from the cone in Proposition~\ref{prop:cone}.
    This represents the completed outcome from a state $x$ using a strategy $\sigma$.
    Explicitly, from Proposition \ref{prop:(co)limit-coincide}, we have:
    \[ {\sf plays}^\sigma_c(x)=\bigvee_{n\in\omega}(1\xrightarrow{K({\sf plays}^\sigma_n(x))} \widebar{H^n_X}(X)
        \xrightarrow{\widebar{H^n_X}(!)} \widebar{H^n_X}(0) \xrightarrow{\widebar{\kappa_n}} (XA)^*XB
    ) \]
    where $\kappa_n:H^n_X(0)\rightarrow (XA)^*XB$ is the coprojection
    from the initial chain for $H_X$ into the initial $H_X$-algebra.
\end{definition}

\begin{remark}
    \label{finite-approximation-remark}
    We stress that the outcome of a strategy from a state is a map
    into $\widetilde{PT}((XA)^*XB)$, rather than $T((XA)^*XB)$.
    The set
    ${\sf plays}^\sigma_c(x)$ will be a singleton if all
    plays which conform to $\sigma$ are completed in a finite
    number of steps, and will otherwise be empty.
    For example, consider the game below with a non-deterministic environment, where there is only one strategy.
    
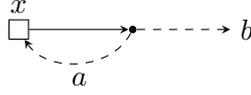
\begin{center}
    \begin{tikzpicture}
        \node[draw] at (0,0) (a0) {};
         \node[circle,fill,inner sep=1pt] at (1.5,0)              (a1) {};
         \node                             at (3,0) (a2) {$b$};
        \draw[-stealth] (a0) -- (a1);
        \draw[-stealth, dashed] (a1) -- (a2) {};
        \draw[-stealth, dashed] (a1) to[bend left=60] node[midway, below] {$a$} (a0) ;
        \node at (0,0.3) {$x$};
    \end{tikzpicture}
    \end{center}
    %\end{wrapfigure}
    We compute the $n$-step partial outcomes
    of the strategy, which include incomplete plays, as $\{x\},\{b,xax\},\{b,xaxb,xaxax\},\cdots$.
    The set ${\sf plays}^\sigma_c(x)$ will however be empty,
    as none of the $n$-step partial outcomes will ever be a set of completed plays
    (i.e. the strategy is not finitely completing).
    In Definition~\ref{def:strategy_outcome},
    the map $\widebar{H^n_X}(!)$ performs this operation
    of removing sets which contain incomplete plays,
    this is discussed further in Proposition~\ref{prop:filter_incomplete}.
    This example also points to another reason why defining the outcome of a strategy $\sigma$ should be done in ${\bf Kl}(\widetilde{PT})$ and not in ${\bf Kl}(T)$: unlike $\widebar{H^n_X}(!)$, if a
    similar map $\widehat{H^n_X}(!)$ existed in ${\bf Kl}(T)$, it would remove the incomplete plays. So even if joins of chains of morphisms did exist in ${\bf Kl}(T)$, the resulting joins would describe the outcomes that the \emph{environment} could force when playing against $\sigma$, and not the set of completed outcomes which the controller can force when playing $\sigma$.
\end{remark}

\section{Executions via Strategies}
\label{sect:executions_via_strategies}

\begin{wrapfigure}{r}{0cm}
\begin{tikzcd}
        &[-1em] X \ar[d, dashed]
            \ar[ddl, bend right, "\widebar{H^n_X}(!)\,\odot\, c^*_n"']
            \ar[ddr, bend left, "\widebar{H^{n+1}_X}(!)\,\odot\, c^*_{n+1}"]
        &[-1em]
        \\[-1em]
        & (XA)^*XB \ar[dl, "\pi_n"'] \ar[dr, "\pi_{n+1}"]
        \\[-1em]
        \widebar{H^n_X}(0) & & \ar[ll, "\widebar{H^n_X}(!)"] \widebar{H^{n+1}_X}(0)
\end{tikzcd}
\end{wrapfigure}

We are now ready to prove our main result,
Theorem~\ref{theorem:main},
which states that the executions from a state $x$
are precisely the union of the completed outcomes of all strategies from $x$.
The proof is spread over the three next lemmas.
The goal is to show
that the map which sends a state to
the union of the outcomes of all strategies for that state, depicted using a dashed arrow, makes
the diagram on the right commute.
The limiting property of $(XA)^*XB$ (recall that this is the limit of the
final sequence of $\widebar{H_X}$) then equates this map with the execution map.
Let $T$ be either $Q$ or $D$, and fix an arbitrary $\widetilde{PT}H$-coalgebra
$(X,c)$ for the remainder of the section.

Our first lemma describes how $n$-step partial outcomes of $n$-step strategies
from $x$
correspond to elements of $c^*_n(x)\in\widetilde{PT}H^n_X(X)$.
It is important we only consider $n$-step strategies rather than full
strategies here, as a deadlock in the future prevents some $\sigma_{n+i}$
from existing (see the comments under Definition~\ref{def:strategy}).
This lemma is an inductive version of what happens
when applying the weak distributives once:
where they reverse the branching by giving the one-step
outcome of each of the controller choices.

\begin{lemma}
    \label{lemma:main}
    Unfolding $c^*$ $n$ times at a state $x$ gives the set of
    $n$-step partial outcomes of all $n$-step strategies starting at $x$.
    \[ c^*_n(x)=\{{\sf plays}^\sigma_n(x)\mid\sigma\in\Sigma_{n,c}(x)\}\]
\end{lemma}
\begin{proof}(Sketch)
    The proof proceeds by induction.
    The base case $(n=0)$ holds as both sides are $\{\eta^T(x)\}$.
    Now we provide a sketch the inductive case for $T=Q$.
    To simplify notation,
    we assume that ${\sf plays}^\sigma_n(x)=\{\rho^\sigma_1x^\sigma_1,
    \rho^\sigma_2x^\sigma_2\}$ for every $\sigma\in\Sigma_{n,c}(x)$,
    i.e. that each pointed $n$-step strategy $(x,\sigma)$ forces a set
    of two partial plays after $n$ steps.
    A similar argument works in general: we provide the details in the \hyperref[appendix:proof:prop:main]{Appendix}.
    \begin{align*}
        c^*_{n+1}(x)
        &=\widebar{H^n_X}(^*c)\odot c^*_n(x)
        \\&
        =\mu^P\mu^Q\circ P\delta^{PQ}\circ PQ(\lambda_n)\circ H_X^n(c)\circ c^*_n(x)
        \\
        &
        =\mu^P\mu^Q\circ P\delta^{PQ}\circ PQ(\lambda_n)\circ H_X^n(c)(
        \{\{\rho^\sigma_1x^\sigma_1,\rho^\sigma_2x^\sigma_2\}
            \mid\sigma\in\Sigma_{n,c}(x)\})
        \\
        &
        =\mu^P\mu^Q(\{\delta^{PQ}(\{\{\{\rho^\sigma_ix^\sigma_i\cdot U\}\mid U\in c^*(x^\sigma_i)\}\mid i\in 1,2\})
            \mid \sigma\in\Sigma_{n,c}(x)\})
        \\
        &
        =
        \mu^P\mu^Q(\{\{\mc{V}^\sigma_1\cup\mc{V}^\sigma_2\mid
                \mc{V}_i^\sigma\subseteq \{\{\rho^\sigma_ix^\sigma_i\cdot U\}
                \mid U\in c^*(x)\}\mt{ for } i\in 1,2\}
            \mid\sigma\in\Sigma_{n,c}(x)\})
        \\
        &
        =
        \{\textstyle\bigcup(\mc{V}^\sigma_1\cup\mc{V}^\sigma_2)\mid
                \sigma\in\Sigma_{n,c}(x),\,
                \mc{V}_i^\sigma\subseteq \{\{\rho^\sigma_ix^\sigma_i\cdot U\}
                \mid U\in c^*(x)\}\mt{ for } i\in 1,2
            \}
    \end{align*}
    The equality after the first line follows by assumption,
    and the inductive hypothesis.
    We also find that:
    \[ \{{\sf plays}^\sigma_{n+1}(x)\mid\sigma\in\Sigma_{n+1,c}(x)\}
    =\{\sigma_{n+1}(\rho^\sigma_1x^\sigma_1)\cup\sigma_{n+1}(\rho^\sigma_2x^\sigma_2)
    \mid \sigma\in\Sigma_{n+1,c}(x)\}\]
    We now show these two sets are equal. $(\supseteq)$
    Assume some $(n+1)$-step partial strategy $\sigma\in\Sigma_{n+1,c}(x)$.
    We immediately have an $n$-step partial strategy $\{\sigma_i\}_{i\le n}$
    (and we know it forces $\rho^\sigma_1x^\sigma_1$ and $\rho^\sigma_2x^\sigma_2$
    by assumption).
    Choosing $\mc{V}_1^\sigma:=\{\sigma_{n+1}(\rho^\sigma_1x^\sigma_1)\}$
    and $\mc{V}_2^\sigma:=\{\sigma_{n+1}(\rho^\sigma_2x^\sigma_2)\}$,
    recovers the element
    $\bigcup(\mc{V}^\sigma_1\cup\mc{V}^\sigma_2)=
    \sigma_{n+1}(\rho^\sigma_1x^\sigma_1)\cup\sigma_{n+1}(\rho^\sigma_1x^\sigma_1)$ in $c^*_{n+1}(x)$.
    For the
    $(\subseteq)$ direction,
    assume some $n$-step partial strategy $\sigma\in\Sigma_{n,c}(x)$,
    and appropriate sets of subsets $\mc{V}^\sigma_i$.
    We can extend $\sigma$ to a $(n+1)$-step strategy
    by choosing $\sigma_{n+1}(\rho^\sigma_1x^\sigma_1):=\bigcup\mc{V}^\sigma_1$
    and $\sigma_{n+1}(\rho^\sigma_2x^\sigma_2):=\bigcup\mc{V}^\sigma_2$.
    This gives us an element in the RHS, corresponding
    to each element in $c^*_{n+1}(X)$ because
    $\bigcup(\mc{V}^\sigma_1\cup\mc{V}^\sigma_2)=\bigcup\mc{V}^\sigma_1\cup
    \bigcup\mc{V}^\sigma_2$.
    Checking this is a well-defined $(n+1)$-step partial strategy
    follows from
    $\mc{V}^\sigma_i\subseteq\{\{\rho^\sigma_ix^\sigma_i\cdot U\}\mid
        U\in c^*(x)\}$.
\end{proof}

Recall that $!:X\rightarrow\widetilde{PT}(0)$ is the unique morphism
into the final object in ${\bf Kl}(\widetilde{PT})$.
The next lemma describes how composing with $\widebar{H^n_X}(!)$
equates the set of all the $n$-step partial outcomes of strategies
in $\Sigma_n(x)$, and the set of all the $n$-step partial outcomes of whole
strategies in $\Sigma_c(x)$.
The map
$\widebar{H^n_X}(!):(XA)^{<n}XB+(XA)^nX\rightarrow (XA)^{<n}XB$
in ${\bf Kl}(\widetilde{PT})$
removes collections which contain incomplete plays.
By direct calculation we can see:
\[ \widebar{H^n_X}(!)\odot U=\bigcup_{u\in U}(P\mu^T\circ\delta^{PT}
    \circ T(\lambda_n\circ H^n_X(!))(u))\]
Now
$\widebar{H^n_X}(!)=\lambda_n\circ H^n_X(!):H^n_X(X)\rightarrow \widetilde{PT}H^n_X(0)$
maps a completed execution $\chi$ to $\{\eta^T(\chi)\}$
and an incomplete execution $\rho$ to $\emptyset$.
Both laws $\delta^{PT}:TP\rightarrow PT$ we use the property that
$\delta(u)=\emptyset\iff \emptyset\in{\sf supp}\, u$.
Thus, in the expression above, only $u$'s which only contain
completed traces will be kept.
We summarise our findings in the following proposition, which will be used to prove Lemma \ref{lemma:equate_n_deep} below.

\begin{proposition}
    \label{prop:filter_incomplete}
    Composition with
    $\widebar{H^n_X}(!):H^n_X(X)\rightarrow\widetilde{PT}H^n_X(0)$
    in ${\bf Kl}(\widetilde{PT})$
    filters out collections (elements of $TH^n_X(X)$) which contain incomplete executions.
    Formally: given some $U\in\widetilde{PT}H^n_X(X)$,
    \[ \widebar{H^n_X}(!)\odot U=\{u\in TH^n_X(X)\mid u\in U\mt{ and }
    {\sf supp}\,u \subseteq H^n_X(0)\} \]
\end{proposition}

\begin{lemma}
    \label{lemma:equate_n_deep}
    The $n$-step partial outcomes of $n$-step strategies which only contain completed
    plays, are equal to the $n$-step partial outcomes of full strategies which only
    contain completed plays.
    For all $x\in X$, we have
    \[\widebar{H^n_X}(!)\odot \{{\sf plays}^\sigma_n(x)\mid \sigma\in\Sigma_{n,c}(x)\}
        =\widebar{H^n_X}(!)\odot\{{\sf plays}^\sigma_n(x)\mid \sigma\in\Sigma_c(x)\}
    \]
\end{lemma}
\begin{proof}
    $(\subseteq)$
    Let $\sigma\in\Sigma_{n,c}(x)$,
    note every trace in the $n$-step outcome is completed by Proposition \ref{prop:filter_incomplete}.
    To extend $\sigma$, we take
    %$\sigma_{n+i+1}:=(\mt{Im}(\sigma_{n+i})\rightarrowtail H^{n+i}_X(0)
    %\rightarrowtail H^{n+i+1}_X(X)\xrightarrow{\eta^T}TH^{n+i+1}_X(X))$
    %for all $i\in\omega$
    $\sigma_{n+i}:{\sf Im}(\sigma_{n+i})\to TH^{n+i+1}_X(X)$ to map each completed trace
    $\chi$ to $\eta^T(\chi)$.
    This has no impact on the $n$-step partial outcome.
    $(\supseteq)$ Let $\sigma\in\Sigma_c(x)$,
    we immediately see that $\{\sigma_i\}_{i\le n}\in\Sigma_{n,c}(x)$,
    and both will result in the same $n$-step partial outcome.
\end{proof}

Now for the final lemma.
We require the
map $\pi_n:(XA)^*XB\rightarrow \widetilde{PT}H^n_X(0)$, which is the unique
projection of the embedding $\widebar{\kappa_n}$ which is lifted from
${\bf Set}$. It is easily calculated as sending %an element 
$\chi\in (XA)^*XB$ to
$\eta^{\widetilde{PT}}(\chi)$ if $\chi\in H^n_X(0)$ and $\emptyset$ otherwise.
Composition with $\pi$ in ${\bf Kl}(\widetilde{PT})$ thus 
filters out collections which contain
complete executions which have length greater than $n$
(are not elements of $H^n_X(0)$).

\[ \pi_n\odot U=\{u\in TH^n_X(0)\mid u \in U \} \]

\begin{lemma}
    \label{lemma:complete_n_deep_outcome}
    The union of the outcomes of strategies which
    only contain completed plays of length less than $n$,
    equals the set of $n$-step partial outcomes of strategies
    which only contain completed plays.
    \[ \pi_n\odot \bigcup_{\sigma\in\Sigma_c(x)}{\sf plays}^\sigma_c(x)
    =\widebar{H^n_X}(!)\odot\{{\sf plays}^\sigma_n(x)\mid\sigma\in\Sigma_c(x)\} \]
\end{lemma}

\begin{theorem}[Executions are Strategies]
    \label{theorem:main}
    The execution map recovers, at a state $x\in X$,
    the union of the completed outcomes of all the strategies
    originating from $x$.
    \[{\sf exec}_c(x)=\bigcup_{\sigma\in\Sigma_c(x)}{\sf plays}^\sigma_c(x)\]
\end{theorem}
\begin{proof} We prove the triangle commutes discussed at the top of the section:
    \begin{align*}
        \widebar{H^n_X}(!)\odot c^*_n(x)
        &=
        \widebar{H^n_X}(!)\odot\bigcup_{\sigma\in\Sigma_n(x)}K({\sf plays}^\sigma_n(x))
        && \mt{Lemma~\ref{lemma:main}}
        \\
        &=
        \widebar{H^n_X}(!)\odot\bigcup_{\sigma\in\Sigma_c(x)}K({\sf plays}^\sigma_n(x))
        && \mt{Lemma~\ref{lemma:equate_n_deep}}
        \\
        &=
        \pi_n\odot\bigcup_{\sigma\in\Sigma_c(x)}{\sf plays}^\sigma_c(x)
        && \mt{Lemma~\ref{lemma:complete_n_deep_outcome}}
    \end{align*}
    which implies that
    $\bigcup_{\sigma\in\Sigma_c(-)}{\sf plays}^\sigma_c(-):X\rightarrow(XA)^*XB$
    is the unique map by the limiting property of $(XA)^*XB$
    with projections $\pi_n:(XA)^XB\rightarrow\widetilde{PT}H^n_X(X)$.
\end{proof}

Recall $f_{\pi_2}:(XA)^*XB\rightarrow A^*B$ from Proposition~\ref{prop:traces_via_executions}.
\begin{corollary}[Traces via Strategies]
    \label{trace_via_strat}
    $\ds{\sf tr}_c(x)=f_{\widebar{\pi_2}}\odot \bigcup_{\sigma\in\Sigma_c(x)}{\sf plays}^\sigma_c(x)$
\end{corollary}

\begin{remark}

Our Theorem \ref{theorem:main} can
be phrased as saying,
that for any $U\subseteq (XA)^*XB$, we have:
\[U\in{\sf exec}_c(x)\iff
\mt{there exists a strategy which forces }U\]
Similarly, Corollary \ref{trace_via_strat} says
that for any $U\subseteq A^*B$, we have:
\[ U\in{\sf tr}_c(x)\iff
\mt{there exists a strategy which forces }U
\]
This allows us to compute the trace semantics
of games by enumerating the strategies:
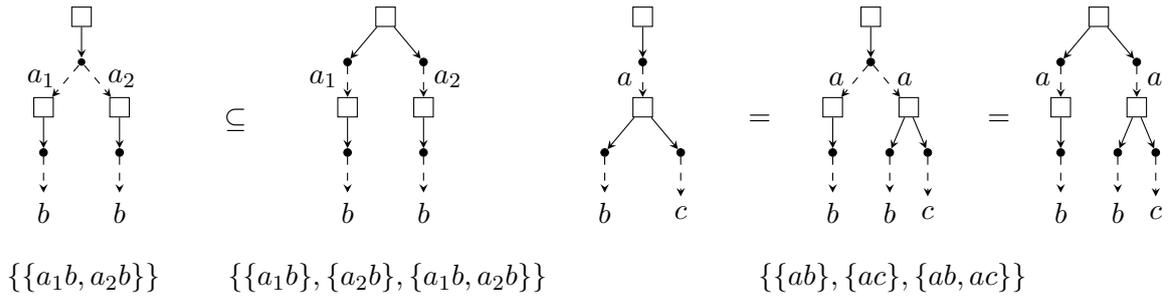
\begin{center}
\vspace{2em}
\tikz[overlay]{
            \node at (1,-0.6) {$\{\{a_1b,a_2b\}\}$};
            \node at (5,-0.6) {$\{\{a_1b\},\{a_2b\},\{a_1b,a_2b\}\}$};
            \node at (3, 1.5) {$\subseteq$};
        }
\begin{tikzpicture}
        % First game
        \node[draw]                           at (0,0.6)     (c0)  {};
        \node[circle,inner sep=1pt,fill] at (0,0)       (e0)  {};
        \node[draw]                           at (-0.5,-0.6) (c10) {};
        \node[draw]                           at (0.5,-0.6)  (c11) {};
        \node[draw,circle,inner sep=1pt,fill] at (-0.5,-1.2) (e10) {};
        \node[draw,circle,inner sep=1pt,fill] at (0.5,-1.2)  (e11) {};
        \node                                 at (-0.5,-2)   (t0) {$b$};
        \node                                 at (0.5,-2)    (t1) {$b$};

        \draw[-stealth]        (c0) -- (e0);
        \draw[-stealth,dashed] (e0) -- (c10) node[left,midway,yshift=1px]  {$a_1$};
        \draw[-stealth,dashed] (e0) -- (c11) node[right,midway,yshift=1px] {$a_2$};
        \draw[-stealth]        (c10) -- (e10);
        \draw[-stealth]        (c11) -- (e11);
        \draw[-stealth,dashed] (e10) -- (t0);
        \draw[-stealth,dashed] (e11) -- (t1);

        % Second game
        \node[draw]                           at (4,0.6)    (c0)  {};
        \node[draw,circle,inner sep=1pt,fill] at (3.5,0)    (e0)  {};
        \node[draw,circle,inner sep=1pt,fill] at (4.5,0)    (e1)  {};
        \node[draw]                           at (3.5,-0.6) (c10) {};
        \node[draw]                           at (4.5,-0.6) (c11) {};
        \node[draw,circle,inner sep=1pt,fill] at (3.5,-1.2) (e10) {};
        \node[draw,circle,inner sep=1pt,fill] at (4.5,-1.2) (e11) {};
        \node                                 at (3.5,-2)   (t0) {$b$};
        \node                                 at (4.5,-2)   (t1) {$b$};

        \draw[-stealth]        (c0) -- (e0);
        \draw[-stealth]        (c0) -- (e1);
        \draw[-stealth,dashed] (e0) -- (c10) node[left,midway,yshift=1px]  {$a_1$};
        \draw[-stealth,dashed] (e1) -- (c11) node[right,midway,yshift=1px] {$a_2$};
        \draw[-stealth]        (c10) -- (e10);
        \draw[-stealth]        (c11) -- (e11);
        \draw[-stealth,dashed] (e10) -- (t0);
        \draw[-stealth,dashed] (e11) -- (t1);

    \end{tikzpicture}\hspace{4em}%
        \tikz[overlay]{
            \node at (4,-0.6) {$\{\{ab\},\{ac\},\{ab,ac\}\}$};
            \node at (2.25, 1.5) {$=$};
            \node at (5.4, 1.5) {$=$};
        }%
    \begin{tikzpicture}
        % First game
        \node[draw]                           at (0,0.6)     (c0)  {};
        \node[draw,circle,inner sep=1pt,fill] at (0,0)       (e0)  {};
        \node[draw]                           at (0,-0.6)    (c1)  {};
        \node[draw,circle,inner sep=1pt,fill] at (-0.5,-1.2) (e10) {};
        \node[draw,circle,inner sep=1pt,fill] at (0.5,-1.2)  (e11) {};
        \node                                 at (-0.5,-2)   (t0)  {$b$};
        \node                                 at (0.5,-2)    (t1)  {$c$};

        \draw[-stealth]        (c0)  -- (e0);
        \draw[-stealth,dashed] (e0)  -- (c1) node[left,midway,yshift=1px]  {$a$};
        \draw[-stealth]        (c1)  -- (e10);
        \draw[-stealth]        (c1)  -- (e11);
        \draw[-stealth,dashed] (e10) -- (t0);
        \draw[-stealth,dashed] (e11) -- (t1);

        % Second game
        \node[draw]                           at (3,0.6)     (c0)  {};
        \node[draw,circle,inner sep=1pt,fill] at (3,0)       (e0)  {};
        \node[draw]                           at (2.5,-0.6)  (c10) {};
        \node[draw]                           at (3.5,-0.6)  (c11) {};
        \node[draw,circle,inner sep=1pt,fill] at (2.5,-1.2)  (e10) {};
        \node[draw,circle,inner sep=1pt,fill] at (3.25,-1.2) (e11) {};
        \node[draw,circle,inner sep=1pt,fill] at (3.75,-1.2) (e12) {};
        \node                                 at (2.5,-2)    (t0)  {$b$};
        \node                                 at (3.25,-2)   (t1)  {$b$};
        \node                                 at (3.75,-2)   (t2)  {$c$};

        \draw[-stealth]        (c0)   -- (e0);
        \draw[-stealth,dashed] (e0)   -- (c10) node[left,midway,yshift=1px]  {$a$};
        \draw[-stealth,dashed] (e0)   -- (c11) node[right,midway,yshift=1px] {$a$};
        \draw[-stealth]        (c10)  -- (e10);
        \draw[-stealth]        (c11)  -- (e11);
        \draw[-stealth]        (c11)  -- (e12);
        \draw[-stealth,dashed] (e10)  -- (t0);
        \draw[-stealth,dashed] (e11)  -- (t1);
        \draw[-stealth,dashed] (e12)  -- (t2);

        % Third game
        \node[draw]                           at (6,0.6)     (c0)  {};
        \node[draw,circle,inner sep=1pt,fill] at (5.5,0)     (e00)  {};
        \node[draw,circle,inner sep=1pt,fill] at (6.5,0)     (e01)  {};
        \node[draw]                           at (5.5,-0.6)  (c10) {};
        \node[draw]                           at (6.5,-0.6)  (c11) {};
        \node[draw,circle,inner sep=1pt,fill] at (5.5,-1.2)  (e10) {};
        \node[draw,circle,inner sep=1pt,fill] at (6.25,-1.2) (e11) {};
        \node[draw,circle,inner sep=1pt,fill] at (6.75,-1.2) (e12) {};
        \node                                 at (5.5,-2)    (t0)  {$b$};
        \node                                 at (6.25,-2)   (t1)  {$b$};
        \node                                 at (6.75,-2)   (t2)  {$c$};

        \draw[-stealth]        (c0)   -- (e00);
        \draw[-stealth]        (c0)   -- (e01);
        \draw[-stealth,dashed] (e00)  -- (c10) node[left,midway,yshift=1px]  {$a$};
        \draw[-stealth,dashed] (e01)  -- (c11) node[right,midway,yshift=1px] {$a$};
        \draw[-stealth]        (c10)  -- (e10);
        \draw[-stealth]        (c11)  -- (e11);
        \draw[-stealth]        (c11)  -- (e12);
        \draw[-stealth,dashed] (e10)  -- (t0);
        \draw[-stealth,dashed] (e11)  -- (t1);
        \draw[-stealth,dashed] (e12)  -- (t2);
    \end{tikzpicture}
    \vspace{3em}
\end{center}

The left example shows why convexity
is a very natural assumption.
The only difference between
in the two games is the controller decides
on the right, whereas
the environment chooses on the left:
thus we should have a trace inclusion,
as the right game is strictly better
for the controller.
Convexity insists on this inclusion,
as there is a convex choice on the right game,
when the controller doesn't pick between
going left or right.
We comment on how we could drop convexity,
and other possible variants, in the next subsection.

The right example demonstrates that different
strategies can result in the same set of traces
in the trace map. There are 3 strategies
on the left and middle game,
but a total of 7 strategies on the right game.
In contrast, for executions, the elements of the execution map {\it are} in bijective correspondence with finitely completing strategies.

\end{remark}

\subsection{Variations}
\begin{wrapfigure}{r}{0cm}
\begin{tikzpicture}
    \node[draw]                           at (0,0.7)     (c0) {};
    \node[draw,circle,inner sep=1pt,fill] at (0,0)    (e0) {};
    \node                                 at (0,-0.8) (c1) {$b$};
    \draw[-stealth]         (c0) -- (e0);
    \draw[-stealth, dashed] (e0) -- (c1);
\end{tikzpicture}\hspace{1em}
\begin{tikzpicture}
    \node[draw]                           at (0,0.7)     (c0) {};
    \node[draw,circle,inner sep=1pt,fill] at (-0.5,0)    (e0) {};
    \node[draw,circle,inner sep=1pt,fill] at (0.5,0)     (e1) {};
    \node                                 at (-0.5,-0.8) (c1) {$b$};
    \node                                 at (0.5,-0.8)  (c2) {$c$};

    \draw[-stealth]         (c0) -- (e0);
    \draw[-stealth]         (c0) -- (e1);
    \draw[-stealth, dashed] (e1) -- (c1);
    \draw[-stealth, dashed] (e0) -- (c1);
    \draw[-stealth, dashed] (e1) -- (c2);
\end{tikzpicture}
\hspace{1em}
\end{wrapfigure}
When modelling a two-player game
with a nondeterministic environment,
the choice of monad affects
how fine or coarse the induced trace equivalence
relation is.
For example,
using the monotone neighbourhood monad (assuming the prerequisites were established to obtain finite trace semantics),
where sets of subsets are upwards closed,
would yield a strictly coarser trace equivalence relation.
The trace semantics of the two games on the right would be identified
if we were to use the monotone neighbourhood monad.

We also highlight that if one wanted to drop convexity,
then a suitable combination would be the powerset monad $P$
with the finite multiset monad $M$.
There exists a distributive law with the analogous
one-step strategy picking behaviour which we rely on.
This would yield a strictly finer trace equivalence relation.
It is also more expressive in that 
we can have multiplicities on environment transitions which represent
how many ways the environment can reach a successor.
We illustrate this by adapting a previous example to use multiset semantics instead.
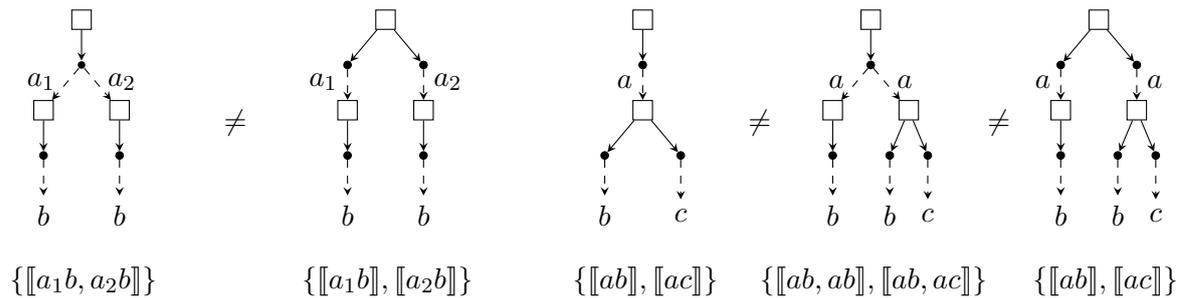
\begin{center}
\vspace{2em}
\tikz[overlay]{
            \node at (1,-0.6) {$\{\llbracket a_1b,a_2b\rrbracket\}$};
            \node at (5,-0.6) {$\{\llbracket a_1b\rrbracket,
                \llbracket a_2b\rrbracket\}$};
            \node at (3, 1.5) {$\ne$};
        }
\begin{tikzpicture}
        % First game
        \node[draw]                           at (0,0.6)     (c0)  {};
        \node[circle,inner sep=1pt,fill] at (0,0)       (e0)  {};
        \node[draw]                           at (-0.5,-0.6) (c10) {};
        \node[draw]                           at (0.5,-0.6)  (c11) {};
        \node[draw,circle,inner sep=1pt,fill] at (-0.5,-1.2) (e10) {};
        \node[draw,circle,inner sep=1pt,fill] at (0.5,-1.2)  (e11) {};
        \node                                 at (-0.5,-2)   (t0) {$b$};
        \node                                 at (0.5,-2)    (t1) {$b$};

        \draw[-stealth]        (c0) -- (e0);
        \draw[-stealth,dashed] (e0) -- (c10) node[left,midway,yshift=1px]  {$a_1$};
        \draw[-stealth,dashed] (e0) -- (c11) node[right,midway,yshift=1px] {$a_2$};
        \draw[-stealth]        (c10) -- (e10);
        \draw[-stealth]        (c11) -- (e11);
        \draw[-stealth,dashed] (e10) -- (t0);
        \draw[-stealth,dashed] (e11) -- (t1);

        % Second game
        \node[draw]                           at (4,0.6)    (c0)  {};
        \node[draw,circle,inner sep=1pt,fill] at (3.5,0)    (e0)  {};
        \node[draw,circle,inner sep=1pt,fill] at (4.5,0)    (e1)  {};
        \node[draw]                           at (3.5,-0.6) (c10) {};
        \node[draw]                           at (4.5,-0.6) (c11) {};
        \node[draw,circle,inner sep=1pt,fill] at (3.5,-1.2) (e10) {};
        \node[draw,circle,inner sep=1pt,fill] at (4.5,-1.2) (e11) {};
        \node                                 at (3.5,-2)   (t0) {$b$};
        \node                                 at (4.5,-2)   (t1) {$b$};

        \draw[-stealth]        (c0) -- (e0);
        \draw[-stealth]        (c0) -- (e1);
        \draw[-stealth,dashed] (e0) -- (c10) node[left,midway,yshift=1px]  {$a_1$};
        \draw[-stealth,dashed] (e1) -- (c11) node[right,midway,yshift=1px] {$a_2$};
        \draw[-stealth]        (c10) -- (e10);
        \draw[-stealth]        (c11) -- (e11);
        \draw[-stealth,dashed] (e10) -- (t0);
        \draw[-stealth,dashed] (e11) -- (t1);

    \end{tikzpicture}\hspace{4em}%
        \tikz[overlay]{
            \node at (0.75,-0.6) {$\{\llbracket ab\rrbracket,\llbracket ac\rrbracket\}$};
            \node at (3.75,-0.6) {$\{\llbracket ab,ab\rrbracket,
                    \llbracket ab,ac\rrbracket\}$};
            \node at (6.8,-0.6) {$\{\llbracket ab\rrbracket,
                    \llbracket ac\rrbracket\}$};
            \node at (2.25, 1.5) {$\ne$};
            \node at (5.4, 1.5) {$\ne$};
        }%
    \begin{tikzpicture}
        % First game
        \node[draw]                           at (0,0.6)     (c0)  {};
        \node[draw,circle,inner sep=1pt,fill] at (0,0)       (e0)  {};
        \node[draw]                           at (0,-0.6)    (c1)  {};
        \node[draw,circle,inner sep=1pt,fill] at (-0.5,-1.2) (e10) {};
        \node[draw,circle,inner sep=1pt,fill] at (0.5,-1.2)  (e11) {};
        \node                                 at (-0.5,-2)   (t0)  {$b$};
        \node                                 at (0.5,-2)    (t1)  {$c$};

        \draw[-stealth]        (c0)  -- (e0);
        \draw[-stealth,dashed] (e0)  -- (c1) node[left,midway,yshift=1px]  {$a$};
        \draw[-stealth]        (c1)  -- (e10);
        \draw[-stealth]        (c1)  -- (e11);
        \draw[-stealth,dashed] (e10) -- (t0);
        \draw[-stealth,dashed] (e11) -- (t1);

        % Second game
        \node[draw]                           at (3,0.6)     (c0)  {};
        \node[draw,circle,inner sep=1pt,fill] at (3,0)       (e0)  {};
        \node[draw]                           at (2.5,-0.6)  (c10) {};
        \node[draw]                           at (3.5,-0.6)  (c11) {};
        \node[draw,circle,inner sep=1pt,fill] at (2.5,-1.2)  (e10) {};
        \node[draw,circle,inner sep=1pt,fill] at (3.25,-1.2) (e11) {};
        \node[draw,circle,inner sep=1pt,fill] at (3.75,-1.2) (e12) {};
        \node                                 at (2.5,-2)    (t0)  {$b$};
        \node                                 at (3.25,-2)   (t1)  {$b$};
        \node                                 at (3.75,-2)   (t2)  {$c$};

        \draw[-stealth]        (c0)   -- (e0);
        \draw[-stealth,dashed] (e0)   -- (c10) node[left,midway,yshift=1px]  {$a$};
        \draw[-stealth,dashed] (e0)   -- (c11) node[right,midway,yshift=1px] {$a$};
        \draw[-stealth]        (c10)  -- (e10);
        \draw[-stealth]        (c11)  -- (e11);
        \draw[-stealth]        (c11)  -- (e12);
        \draw[-stealth,dashed] (e10)  -- (t0);
        \draw[-stealth,dashed] (e11)  -- (t1);
        \draw[-stealth,dashed] (e12)  -- (t2);

        % Third game
        \node[draw]                           at (6,0.6)     (c0)  {};
        \node[draw,circle,inner sep=1pt,fill] at (5.5,0)     (e00)  {};
        \node[draw,circle,inner sep=1pt,fill] at (6.5,0)     (e01)  {};
        \node[draw]                           at (5.5,-0.6)  (c10) {};
        \node[draw]                           at (6.5,-0.6)  (c11) {};
        \node[draw,circle,inner sep=1pt,fill] at (5.5,-1.2)  (e10) {};
        \node[draw,circle,inner sep=1pt,fill] at (6.25,-1.2) (e11) {};
        \node[draw,circle,inner sep=1pt,fill] at (6.75,-1.2) (e12) {};
        \node                                 at (5.5,-2)    (t0)  {$b$};
        \node                                 at (6.25,-2)   (t1)  {$b$};
        \node                                 at (6.75,-2)   (t2)  {$c$};

        \draw[-stealth]        (c0)   -- (e00);
        \draw[-stealth]        (c0)   -- (e01);
        \draw[-stealth,dashed] (e00)  -- (c10) node[left,midway,yshift=1px]  {$a$};
        \draw[-stealth,dashed] (e01)  -- (c11) node[right,midway,yshift=1px] {$a$};
        \draw[-stealth]        (c10)  -- (e10);
        \draw[-stealth]        (c11)  -- (e11);
        \draw[-stealth]        (c11)  -- (e12);
        \draw[-stealth,dashed] (e10)  -- (t0);
        \draw[-stealth,dashed] (e11)  -- (t1);
        \draw[-stealth,dashed] (e12)  -- (t2);
    \end{tikzpicture}
    \vspace{3em}
\end{center}

Moreover, one could consider using finite multisets for the controller
and the environment, with the law $MM\to MM$ -
however restricting the controller to have finite branching
would stop us being able to apply Theorem \ref{theorem:traces}. Using
the list monad $L$ for the environment, together with the distributive law $LP\to PL$ would
likely also work;
this would make it possible to express an order over environment transitions.

\section{Conclusion}

We have shown how to fit two-player controller-versus-environment games
into the finite trace semantics framework of Hasuo, Jacobs and Sokolova \cite{hasuo_generic_2007}, by
identifying that the monads $\widetilde{PQ}$ and $\widetilde{PD}$,
built from a weak distributive laws,
and the class of linear functors satisfy
the necessary requirements.
Along the way we uncovered two mistakes in \cite{bonchi_convexity_2022}
and \cite{jacobs_coalgebraic_2008},
which manifest in our work as composition in $\widetilde{PP}$
not being left-strict, and $\widetilde{PQ}$ not being a commutative monad respectively.
We gave a categorical definition of strategies in games,
and showed how the execution map recovers the set of collections of plays
which can be forced by a finitely completing controller strategy.

Building on the automata-theoretic approach to synthesis,
future work will include using a non-deterministic automaton to describe the desired linear-time behaviour of game plays, and using a product construction to aid synthesis.
By our main result, if the set of collections of traces
from a state in the product is non-empty (and if traces on the product give traces on the game),
then the controller has a strategy to realise the desired behaviour.

In the future, we would also like to extend our results to \textit{infinite} traces,
since in the context of verification and synthesis, one is mainly interested in linear-time properties of potentially infinite computations.
Results in \cite{GoubaultLarrecq2024WeakDL} show the
appropriate weak distributive laws exist for monads
modelling continuous probability on certain categories of
topological spaces,
which will be required when moving to infinite traces.
Another possible extension would be incorporating simple stochastic
games into our results,
however work in \cite{aristote} indicates that the natural approach of lifting
the weak distributive law $PP\to PP$ to the category
of convex algebras may not be possible.

\bibliographystyle{entics} % NEED TO CHANGE THIS to ./entics
\bibliography{refs}

\appendix

\paragraph{Proof Details for Proposition~\ref{prop:kl_cpo_enriched}}
    \label{appendix:kl_cpo_enriched}
    \[\mt{(LHS)}=\{\tts\bigcup
        \ds\bigcup_{y\in U}\mc{U}_y\mid U\in f(x),\,
            \emptyset\subset\mc{U}_y\subseteq_\omega
    {\sf cl}_Z\bigcup_{i\in\omega}g_i(y)\}\]
    \[\mt{(RHS)}=
        {\sf cl}_Z\{\tts\bigcup\ds\bigcup_{y\in V}\mid
        V\in f(x),\,i\in\omega,\,\emptyset\subset\mc{V}_y
        \subseteq g_i(y)\}\]
    $(\subseteq)$ Fix $U\in f(x)$ and $\mc{U}_y$
    for each $y\in U$.
    Let $\mc{U}_y=\{\bigcup\mc{U}_{yj}\}_{j\in J}$
    where
    $\emptyset\subset \mc{U}_{yj}\subseteq_\omega\ds\bigcup_{i\in\omega}g_i(y)$
    and $J$ is some finite set.
    Let
    $\mc{A}_{iy}:=\ds\bigcup_{j\in J}\mc{U}_{yj}\cap g_i(y)$.
    Notice $\mc{A}_{0y}\subseteq\mc{A}_{1y}\subseteq\mc{A}_{2y}\dots$
    and
    \[\bigcup_{i\in\omega}\mc{A}_{iy}=
    \bigcup_{i\in\omega}(\bigcup_{j\in J}\mc{U}_{yj}\cap g_i(y))
    =\bigcup_{j\in J}\mc{U}_{yj}\cap\bigcup_{i\in\omega}g_i(y)=
    \bigcup_{j\in J}\mc{U}_{yj}\]
    is finite,
    meaning that for all $y\in U$, we have some $k_y\in\omega$
    with $\mc{A}_{k_yy}=\ds\bigcup_{i\in\omega}\mc{A}_{iy}$.
    We set $k:=\ds\max_{y\in U}k_y$,
    so we have $\bigcup\ds\bigcup_{y\in U}\mc{A}_{ky}\in\mt{(RHS)}$.
    Now we show this is equal to the element in the $(\mt{LHS})$
    we started with.
    \[\tts\bigcup\ds\bigcup_{y\in U}\mc{A}_{ky}=\tts\bigcup
        \ds\bigcup_{y\in U}\bigcup_{j\in J}\mc{U}_{yj}=
        \bigcup_{y\in U}\bigcup_{j\in J}\tts\bigcup\mc{U}_{yj}
        =\bigcup\ds\bigcup_{y\in U}\mc{U}_y
        \]
    $(\supseteq)$ Because the $(\mt{LHS})$ is a convex set,
    it suffices to show
    \[(\mt{LHS})\supseteq
        {\sf cl}_Z\{\tts\bigcup\ds\bigcup_{y\in V}\mid
        V\in f(x),\,i\in\omega,\,\emptyset\subset\mc{V}_y
        \subseteq g_i(y)\}\]
    Suppose some $V\in f(x)$, $i\in\omega$, and
    $\emptyset\subset\mc{V}_y\subseteq g_i(y)$
    for each $y\in V$.
    We have $\mc{V}_y\subseteq_\omega g_i(y)\subseteq_\omega{\sf cl}_Z
    \ds\bigcup_{i\in\omega}g_i(y)$,
    hence
    $\bigcup\ds\bigcup_{y\in V}\mc{V}_y\in(\mt{LHS})$.

    Now
    we show that case for $T=D$.
$(-)^\$_{XY}:\widetilde{PD}(Y)^X\rightarrow\widetilde{PD}(Y)^{\widetilde{PD}(X)}$
\begin{align*}
    f^\$(U)
    &=\mu^P\mu^D\circ P\delta^{PD}\circ PD(f)(U)
    \\
    &=\mu^P\mu^D\circ P\delta^{PD}(\{[f(x_i)\mapsto p_i]
        \mid [x_i\mapsto p_i]\in U\})
    \\
    &=\mu^P\mu^D(\{
        \{\mu^D[\Phi_i\mapsto p_i]\mid {\sf supp}\,\Phi_i\subseteq f(x_i)\}
        \mid [x_i\mapsto p_i]\in U\})
    \\
    &=\{
        \mu^D\mu^D[\Phi_i\mapsto p_i]\mid[x_i\mapsto p_i]\in U,\,
        {\sf supp}\,\Phi_i\subseteq f(x_i)
        \}
\end{align*}

Show that $\odot$ is bilinear.
It preserves joins in it's first argument:
\begin{align*}
    &\bigvee_{n\in\omega} g_n\odot f(x)
    \\
    &=(\bigvee_{n\in\omega}g_n)^\$(f(x))
    \\
    &=\{\mu^D\mu^D[\Phi_i\mapsto p_i]\mid [y_i\mapsto p_i]\in f(x),\
    {\sf supp}\,\Phi_i\subseteq {\sf cl}\bigcup_{n\in\omega} g_n(y_i)\}
    \\
    &\overset{(*)}{=}{\sf cl}
    \{\mu^D\mu^D[\Phi_i\mapsto p_i]\mid[y_i\mapsto p_i]\in f(x),\,n\in\omega,\,
        {\sf supp}\,\Phi_i\subseteq g_n(y_i)\}
    \\
    &={\sf cl}\bigcup_{n\in\omega}
    \{\mu^D\mu^D[\Phi_i\mapsto p_i]\mid[y_i\mapsto p_i]\in f(x),\,
        {\sf supp}\,\Phi_i\subseteq g_n(y_i)\}
    \\
    &={\sf cl}\bigcup_{n\in\omega}g_n^\$(f(x))
    \\
    &=\bigvee_{n\in\omega}(g_n\odot f)(x)
\end{align*}

We prove $\mt{(LHS)}\overset{(*)}{=}{\sf cl}\,\mt{(RHS)}$.
$(*)(\subseteq)$ Fix some $[y_i\mapsto p_i]\in f(x)$.
Take each $\Phi_i$ with
${\sf supp}\,\Phi_i\subseteq{\sf cl}\ds\bigcup_{n\in\omega}g_n(y_i)$.
We have $\Phi_i=[\mu^D(\Phi_{ij})\mapsto p_{ij}]_{j\in J}$
where ${\sf supp}\,\Phi_{ij}\subseteq\ds\bigcup_{n\in\omega}g_n(y_i)$.
From ${\sf supp}\,\Phi_{ij}$ being finite
and $g_n$'s being increasing,
there must exist some $n_{ij}$ such that
${\sf supp}\,\Phi_{ij}\subseteq g_{n_{ij}}(y_i)$.
Choose $n:=\ds\max_{i,j}\,n_{ij}$
so we have
\[ {\sf supp}\,\Phi_{ij}\subseteq g_n(y_i)
    \implies
{\sf supp}\,\mu^D([\Phi_{ij}\mapsto p_{ij}])\subseteq g_n(y_i)\]
Hence we have an element $\mu^D\mu^D[\mu^D([\Phi_{ij}\mapsto p_{ij}])\mapsto p_i]$
of the $\mt{(RHS)}$, which is thus
an element of ${\sf cl}\,(\mt{RHS})$ because ${\sf cl}$ is extensive.
Finally, must show this is equivalent to the element we started with:
\begin{align*}
    \mu^D\mu^D[\mu^D([\Phi_{ij}\mapsto p_{ij}])\mapsto p_i]
    &
    =\mu^D\mu^DD(\mu^D)[[\Phi_{ij}\mapsto p_{ij}]\mapsto p_i]
    \\
    &
    =\mu^D\mu^DDD(\mu^D)[[\Phi_{ij}\mapsto p_{ij}]\mapsto p_i]
    \\
    &
    =\mu^D\mu^D[[\mu^D(\Phi_{ij})\mapsto p_{ij}]\mapsto p_i]
\end{align*}
which follows from
\[ \mu^D\mu^DD(\mu^D)\overset{\mt{assoc.}}=\mu^D\mu^D\mu^D
\overset{\mt{assoc.}}=\mu^DD(\mu^D)\mu^D\overset{\mt{nat.}}
=\mu^D\mu^DDD(\mu^D)\]
$(*)(\supseteq)$ We show $(\mt{LHS})\supseteq(\mt{RHS})$
which implies $(\mt{LHS})={\sf cl}\,(\mt{LHS})\supseteq{\sf cl}\,(\mt{RHS})$
by the $\mt{(LHS)}$ being closed and montonicity of ${\sf cl}$.
Fix $[y_i\mapsto p_i]\in f(x)$, $n\in\omega$,
and some $\Phi_i$ with ${\sf supp}\,\Phi_i\subseteq g_n(y_i)$.
We have ${\sf supp}\,\Phi_i\subseteq\ds\bigcup_{n\in\omega}g_n(y_i)$
so $\mu^D\mu^D[\Phi_i\mapsto p_i]\in(\mt{LHS})$.

Now show that $\odot$ preserves joins in it's second argument:

\begin{align*}
    &g\odot\bigvee_{n\in\omega}f_n(x)
    \\
    &=g^\$({\sf cl}\,\bigcup_{n\in\omega}f_n(x))
    \\
    &=\{\mu^D\mu^D[\Phi_i\mapsto p_i]\mid
        [y_i\mapsto p_i]\in{\sf cl}\,\bigcup_{n\in\omega}f_n(x),\,
        {\sf supp}\,\Phi_i\subseteq g(y_i)\}
    \\
    &\overset{(*)}={\sf cl}
    \{\mu^D\mu^D[\Phi_i\mapsto p_i]\mid n\in\omega,\,
        [y_i\mapsto p_i]\in f_n(x),\,
    {\sf supp}\,\Phi_i\subseteq g(y_i)\}
    \\
    &={\sf cl}\,\bigcup_{n\in\omega}
    \{\mu^D\mu^D[\Phi_i\mapsto p_i]\mid [y_i\mapsto p_i]\in f_n(x),\,
    {\sf supp}\,\Phi_i\subseteq g(y_i)\}
    \\
    &={\sf cl}\,\bigcup_{n\in\omega}g^\$(f_n(x))
    \\
    &=\bigvee_{n\in\omega}(g\odot f_n)(x)
\end{align*}
Where $(*)$ follow from:
\begin{align*}
    \mu^D(\Phi)\in{\sf cl}\bigcup_{n\in\omega}f_n(x)
    &\iff{\sf supp}\,\Phi\subseteq\bigcup_{n\in\omega}f_n(x)
    \\
    &\iff n\in\omega,\, {\sf supp}\,(\Phi)\subseteq f_n(x)
    \\
    &\iff n\in\omega,\, \mu^D(\Phi)\in{\sf cl}\, f_n(x)
    \\
    &\iff n\in\omega,\, \mu^D(\Phi)\in f_n(x)
\end{align*}

\paragraph{Proof Details for Counterexample~\ref{counter:left_strict}}
\label{appendix:left_strict_counter_info}
Compute the Kleisli extension $(-)^\$$ of $\bot:X\rightarrow \widetilde{PP_f}(Y)$.
\begin{align*}
    \bot^\$(\mc{U})
    &=\mu^{\widetilde{PP_f}}\circ \widetilde{PP_f}(\bot)(\mc{U})
    \\
    &=\mu^{\widetilde{PP_f}}\{\{\bot(x)\mid x\in U\}\mid U\in\mc{U}\}
    \\
    &=\mu^P\mu^{P_f}\circ P\delta^{PP_f}\{\{\bot(x)\mid x\in U\}\mid U\in\mc{U}\}
    \\             
    &=\mu^P\mu^{P_f}\circ \{\delta^{PP_f}\{\bot(x)\mid x\in U\}\mid U\in\mc{U}\}
    \\             
    &=\mu^P\mu^{P_f} \{\delta^{PP_f}\{\emptyset\mid x\in U\}\mid U\in\mc{U}\}
    \\             
    &=\mu^P\mu^{P_f} \{\begin{cases}
        \{\emptyset\}& \mt{if }U =\emptyset\\
        \emptyset & \mt{otherwise}
    \end{cases}\mid U\in\mc{U}\}
    \\
    \\
    &=\mu^P \{\begin{cases}
        \{\emptyset\}& \mt{if }U =\emptyset\\
        \emptyset & \mt{otherwise}
    \end{cases}\mid U\in\mc{U}\}
    \\
    &=\begin{cases}
        \{\emptyset\} & \mt{if }\emptyset\in\mc{U}
        \\
        \emptyset & \mt{otherwise}
    \end{cases}
\end{align*}
Hence for $f:\{x\}\rightarrow Y$ which maps $x\mapsto \{\emptyset\}$
\[\bot \odot f(x)=\bot^\$(\{\emptyset\})=\{\emptyset\}\ne \emptyset=\bot(x)\]

\paragraph{Proof Details for Counterexample~\ref{counter:commutativity}}
\label{appendix:commutativity_counter_info}

Recall that for the monad $\widetilde{PQ}$ to be commutative,
the following must commute:

\begin{center}
\begin{tikzcd}
    \widetilde{PQ}(\widetilde{PQ}(X)\times Y)
    \ar[d, "\widetilde{PQ}({\sf str})"]
    &
    \ar[l, "{\sf stl}"']
    \widetilde{PQ}(X) \times \widetilde{PQ}(Y)
    \ar[r, "{\sf str}"]
    &
    \widetilde{PQ}(X\times \widetilde{PQ}(Y))
    \ar[d, "\widetilde{PQ}({\sf stl})"]
    \\
    \widetilde{PQ}\widetilde{PQ}(X\times Y)
    \ar[r, "\mu"]
    &
    \widetilde{PQ}(X\times Y)
    &
    \widetilde{PQ}\widetilde{PQ}(X\times Y)
    \ar[l, "\mu"']
\end{tikzcd}
\end{center}
We can draw the set $\{\{x_1\},\{x_2\},\{x_1,x_2\}\}$
as a game on the left (ignoring sets which arise from convexity), and the set $\{\{y_1\},\{y_2\}\}$ as a game
on the right.

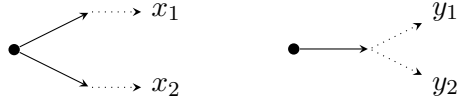
\begin{center}
\begin{tikzpicture}
    \node[circle,fill,inner sep=1.5pt] at (0,0) (n1) {};
    \node[inner sep=0pt] at (1,0.5) (n2) {};
    \node[inner sep=0pt] at (1,-0.5) (n3) {};
    \node at (2,0.5) (n4) {$x_1$};
    \node at (2,-0.5) (n5) {$x_2$};

    \draw[-stealth] (n1) -- (n2);
    \draw[-stealth] (n1) -- (n3);
    \draw[-stealth, dotted] (n2) -- (n4);
    \draw[-stealth, dotted] (n3) -- (n5); 
\end{tikzpicture}\hspace{3em}
\begin{tikzpicture}
    \node[circle,fill,inner sep=1.5pt] at (0,0) (n1) {};
    \node[inner sep=0pt] at (1,0) (n2) {};
    \node at (2,0.5) (n4) {$y_1$};
    \node at (2,-0.5) (n5) {$y_2$};

    \draw[-stealth] (n1) -- (n2);
    \draw[-stealth, dotted] (n2) -- (n4);
    \draw[-stealth, dotted] (n2) -- (n5); 
\end{tikzpicture}
\end{center}
Using the left stength map first and then
the right we get the game:

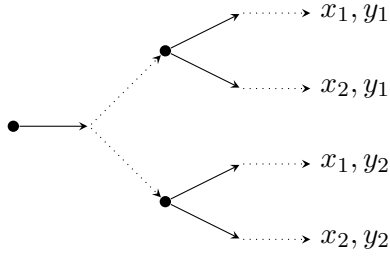
\begin{center}
\begin{tikzpicture}
    \node[circle,fill,inner sep=1.5pt] at (0,0) (n1) {};
    \node[inner sep=0pt] at (1,0) (n2) {};

    \node[circle,fill,inner sep=1.5pt] at (2,1) (m1) {};
    \node[inner sep=0pt] at (3,1.5) (m2) {};
    \node[inner sep=0pt] at (3,0.5) (m3) {};
    \node at (4.5,1.5) (m4) {$x_1,y_1$};
    \node at (4.5,0.5) (m5) {$x_2,y_1$};
    \node[circle,fill,inner sep=1.5pt] at (2,-1) (m1') {};
    \node[inner sep=0pt] at (3,-1.5) (m2') {};
    \node[inner sep=0pt] at (3,-0.5) (m3') {};
    \node at (4.5,-1.5) (m4') {$x_2,y_2$};
    \node at (4.5,-0.5) (m5') {$x_1,y_2$};

    \draw[-stealth] (n1) -- (n2);
    \draw[-stealth, dotted] (n2) -- (m1);
    \draw[-stealth, dotted] (n2) -- (m1'); 
    \draw[-stealth] (m1) -- (m2);
    \draw[-stealth] (m1) -- (m3);
    \draw[-stealth, dotted] (m2) -- (m4);
    \draw[-stealth, dotted] (m3) -- (m5);
    \draw[-stealth] (m1') -- (m2');
    \draw[-stealth] (m1') -- (m3');
    \draw[-stealth, dotted] (m2') -- (m4');
    \draw[-stealth, dotted] (m3') -- (m5');
\end{tikzpicture}
\end{center}
which has enforceable sets ${\sf cl}\{\{x_1y_1,x_1y_2\}, \{x_1y_1,x_2y_2\},\{x_2y_1,x_1y_2\},\{x_2y_1,x_1y_2\}\}$.
When we use the right strength map followed
by the left, we get the game:
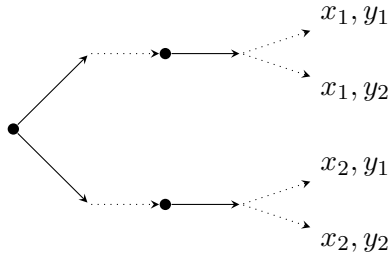
\begin{center}
\begin{tikzpicture}
    \node[circle,fill,inner sep=1.5pt] at (0,0) (n1) {};
    \node[inner sep=0pt] at (1,1) (n2) {};
    \node[inner sep=0pt] at (1,-1) (n3) {};

    \draw[-stealth] (n1) -- (n2);
    \draw[-stealth] (n1) -- (n3);
    \draw[-stealth, dotted] (n2) -- (m1);
    \draw[-stealth, dotted] (n3) -- (m1'); 

    \node[circle,fill,inner sep=1.5pt] at (2,1) (m1) {};
    \node[inner sep=0pt] at (3,1) (m2) {};
    \node at (4.5,1.5) (m4) {$x_1,y_1$};
    \node at (4.5,0.5) (m5) {$x_1,y_2$};

    \draw[-stealth] (m1) -- (m2);
    \draw[-stealth, dotted] (m2) -- (m4);
    \draw[-stealth, dotted] (m2) -- (m5); 

    \node[circle,fill,inner sep=1.5pt] at (2,-1) (m1') {};
    \node[inner sep=0pt] at (3,-1) (m2') {};
    \node at (4.5,-0.5) (m4') {$x_2,y_1$};
    \node at (4.5,-1.5) (m5') {$x_2,y_2$};

    \draw[-stealth] (m1') -- (m2');
    \draw[-stealth, dotted] (m2') -- (m4');
    \draw[-stealth, dotted] (m2') -- (m5'); 
\end{tikzpicture}
\end{center}
which has enforceable sets ${\sf cl}\{\{x_1y_1,x_1y_2\},\{x_2y_1, x_2y_2\}\}$.
Hence:

\begin{align*}
        &\mu^{\widetilde{PQ}}\circ\widetilde{PQ}(\mt{stl})
        \circ\mt{str}(\{\{x_1\},\{x_2\},\{x_1,x_2\},\{\{y_1,y_2\}\})
            \\
        &={\sf cl}_{X\times Y}\{\{x_1y_1,x_1y_2\},\{x_2y_1,x_2y_2\}\}
        \\
        &\subset{\sf cl}_{X\times Y}\{\{x_1y_1,x_1y_2\},\{x_2y_1,x_2y_2\},
        \{x_2y_1,x_1y_2\},\{x_1y_1,x_2y_2\}\}
        \\
        &=\mu^{\widetilde{PQ}}\circ\widetilde{PQ}(\mt{str})
        \circ\mt{stl}(\{\{x_1\},\{x_2\}\},\{y_1,y_2\}\})
\end{align*}

\paragraph{Proof Details for Lemma~\ref{prop:beh_factors_nat_trans}}
\label{appendix:details_beh_factors_nat_trans}
\begin{center}
    \begin{tikzcd}
        X \ar[d, "c"]
            \ar[r, "\mt{beh}_c"]
            \ar[rr, bend left=20, "\mt{beh}_{E_\alpha(c)}"]
        &[3em]
        Z_F \ar[r, "f_\alpha"] \ar[d, "\zeta_F", "\rotatebox{90}{$\sim$}"']
        &[3em]
        Z_G \ar[dd, "\zeta_G", "\rotatebox{90}{$\sim$}"']
        \\[-1em]
        F(X) \ar[d, "\alpha_X"]
            \ar[r, "F(\mt{beh}_c)"]
        &
        F(Z_F)
            \ar[d, "\alpha_{Z_F}"]
        \\[-1em]
        G(X) \ar[r, "G(\mt{beh}_c)"]
            \ar[rr, bend right=20, "G(\mt{beh}_{E_\alpha(c)})"']
        &
        G(Z_F) \ar[r, "G(f_\alpha)"]
        &
        G(Z_G)
    \end{tikzcd}
    \end{center}
    $E_\alpha$ is an identity-on-morphisms functor defined as
    mapping an $F$-coalgebra $c:X\rightarrow F(X)$ to
    a $G$-coalgebra $\alpha_X\circ c:X\rightarrow G(X)$.
    That this maps $F$-coalgebra morphisms to $G$-coalgebra morphisms
    follows from naturality of $\alpha$.
    The map $f_\alpha:=\mt{beh}_{H_\alpha(\zeta_F)}$ is obtained from finality.

\paragraph{Proof Details for Lemma~\ref{lemma:main}}
\label{appendix:proof:prop:main}

We prove $\{{\sf plays}^\sigma_n(x)\mid\sigma\in\Sigma_n(x)\}=c^*_n(x)$
for all $x\in X$ by induction.
The base case holds $(n=0)$ as both sides equal $\{\eta^T(x)\}$.
Example both sides for the $n+1$ case:
\begin{align*}
    \mt{(LHS)}
    &=
    \{{\sf plays}^\sigma_{n+1}(x)\mid\sigma\in\Sigma_{n+1}(x)\}
    \\
    &=
    \{\sigma_{n+1}\odot{\sf plays}^\sigma_n(x)\mid\sigma\in\Sigma_{n+1}(x)\}
    \\
    &=
    \{\mu^T\circ T(\sigma_{n+1})\circ
        {\sf plays}^\sigma_n(x)\mid\sigma\in\Sigma_{n+1}(x)\}
\end{align*}
\begin{align*}
    \mt{(RHS)}
    &=c^*_{n+1}(x)
    \\
    &=\widebar{H^n_X}(c)\odot c^*_n(x)
    \\
    &=\mu^P\mu^T\circ P\delta^{PT}\circ{PT}(\lambda_n)
    \circ{PT}H^n_X(c^*)\circ c^*_n(x)
    \\
    &\overset{\mt{(IH)}}
    =\mu^P\mu^T\circ P\delta^{PT}\circ PT(\lambda_n)
    \circ PT H^n_X(c^*)
    (\{{\sf plays}^\sigma_n(x)\mid \sigma\in\Sigma_n(x)\})
    \\
    &=\bigcup_{\sigma\in\Sigma_n(x)}
(P\mu^T\circ\delta^{PT}\circ T(\lambda_n)\circ TH^n_X(c^*)({\sf plays}^\sigma_n(x)))
\end{align*}

For the case $T=Q$, assume ${\sf plays}^\sigma_n(x)=
\{\chi_1^\sigma,\dots,\chi_{m^\sigma}^\sigma,\rho^\sigma_1x^\sigma_1,
\dots,\rho^\sigma_{l^\sigma}x^\sigma_{l^\sigma}\}$
where $\chi^\sigma_i\in (XA)^{<n}XB$ for $0<i\le m^\sigma$
and $\rho_i^\sigma x^\sigma_i\in (XA)^nX$ for $0<i\le l^\sigma$.
We can compute more detailed expressions for each side.
Recall that the definition of a strategy gives that
$\sigma_{n+1}(\chi^\sigma_i)=\chi^\sigma_i$, i.e. it preserves completed plays.
\begin{align*}
    \mt{(LHS)}
    &=
        \{\mu^Q\circ Q(\sigma_{n+1})
        (
        \{\chi_1^\sigma,\dots,\chi_{m^\sigma}^\sigma,\rho^\sigma_1x^\sigma_1,
        \dots,\rho^\sigma_{l^\sigma}x^\sigma_{l^\sigma}\})
        \mid \sigma\in\Sigma_{n+1}(x)\}
    \\
    &=
        \{\mu^Q
        (
        \{\{\chi_1^\sigma\},\dots,\{\chi_{m^\sigma}^\sigma\},
        \sigma_{n+1}(\rho^\sigma_1x^\sigma_1),
        \dots,\sigma_{n+1}(\rho^\sigma_{l^\sigma}x^\sigma_{l^\sigma})\})
        \mid \sigma\in\Sigma_{n+1}(x)\}
    \\
    &=
        \{
        \{\chi_1^\sigma,\dots,\chi_{m^\sigma}^\sigma\}\cup
        \sigma_{n+1}(\rho^\sigma_1x^\sigma_1)\cup
        \dots\cup\sigma_{n+1}(\rho^\sigma_{l^\sigma}x^\sigma_{l^\sigma})
        \mid \sigma\in\Sigma_{n+1}(x)\}
\end{align*}
We introduce shorthand
$\mc{W}^\sigma_i=\{\{\rho^\sigma_ix^\sigma_iu\mid u\in U\}\mid U\in c^*(x^\sigma_i)\}$
\begin{align*}
    (\mt{RHS})
    &=
    P\mu^Q\circ\delta^{PQ}\circ Q(\lambda_n)\circ QH^n_X(c^*)
    (\{\chi_1^\sigma,\dots,\chi_{m^\sigma}^\sigma,\rho^\sigma_1x^\sigma_1,
    \dots,\rho^\sigma_{l^\sigma}x^\sigma_{l^\sigma}\})
    \\
    &=
    P\mu^Q\circ\delta^{PQ}\circ Q(\lambda_n)
    (\{\chi_1^\sigma,\dots,\chi_{m^\sigma}^\sigma,\rho^\sigma_1c^*(x^\sigma_1),
    \dots,\rho^\sigma_{l^\sigma}c^*(x^\sigma_{l^\sigma})\})
    \\
    &=
    P\mu^Q\circ\delta^{PQ}
    (\{\{\{\chi_1^\sigma\}\},\dots,\{\{\chi_{m^\sigma}^\sigma\}\},
        \mc{W}^\sigma_1,\dots,\mc{W}^\sigma_{l^\sigma}\})
    \\
    &=
    P\mu^Q
    (\{
    \{\{\chi_1^\sigma\},\dots,\{\chi_{m^\sigma}^\sigma\}\}\cup
    \bigcup_{i\le l^\sigma}\mc{V}^\sigma_i\mid
    \emptyset\subset\mc{V}^\sigma_i\subseteq_\omega\mc{W}^\sigma_i\})
    \\
    &=
    \{
    \{\chi_1^\sigma,\dots,\chi_{m^\sigma}^\sigma\}\cup
    \tts\bigcup\ds\bigcup_{i\le l^\sigma}\mc{V}^\sigma_i\mid
    \emptyset\subset\mc{V}^\sigma_i\subseteq_\omega\mc{W}^\sigma_i\}
\end{align*}
So $\mt{RHS}=
    \{
    \{\chi_1^\sigma,\dots,\chi_{m^\sigma}^\sigma\}\cup
    \tts\bigcup\ds\bigcup_{i\le l^\sigma}\mc{V}^\sigma_i\mid
    \sigma\in\Sigma_n(x),\,
    \emptyset\subset\mc{V}^\sigma_i\subseteq_\omega\mc{W}^\sigma_i\}$.
We now justify that $(\mt{LHS})=(\mt{RHS})$.

$(\subseteq)$ Assume we have element
$\{\chi^\sigma_1,\dots,\chi^\sigma_{m^\sigma}\}\cup\sigma_{n+1}
(\rho^\sigma_1x^\sigma_1)\cup\dots\cup\sigma_{n+1}
(\rho^\sigma_{l^\sigma}x^\sigma_{l^\sigma})\in\mt{LHS}$ for some
$(n+1)$-deep strategy
$\sigma\in\Sigma_{n+1}(x)$.
Note that we have an $n$-deep strategy
$\sigma':=\{\sigma_i\}_{i\le n}\in\Sigma_n(x)$,
with the same $n$-deep outcome set
${\sf plays}^{\sigma'}_n(x)={\sf plays}^\sigma_n(x)$.
From now on we just think of $\sigma$ as $n$-deep and $(n+1)$-deep
simultaneously.
We pick $\mc{V}^\sigma_i=\{\sigma_{n+1}(\rho^\sigma_ix^\sigma_i)\}$,
we have $\mc{V}_i^\sigma\subseteq_\omega\mc{W}^\sigma_i$ because
by the definition of a strategy we have
$\sigma_{n+1}(\rho x)=\{\rho u\mid u \in U\}\mt{ for some } U\in c^*(x)$,
i.e. that strategies are confined to choosing successors in $c$.
We have
$
    \sigma_{n+1}(\rho^\sigma_1x^\sigma_1)\cup\dots
    \cup\sigma_{n+1}(\rho^\sigma_{l^\sigma}x^\sigma_{l^\sigma})
    =\tts\bigcup\ds\bigcup_{i\le l^\sigma}\mc{V}^\sigma_i$,
so the element we started with is in the $\mt{RHS}$.

$(\supseteq)$ Suppose we have
$\{\chi^\sigma_1,\dots,\chi^\sigma_{m^\sigma}\}\cup
\bigcup\ds\bigcup_{i\le l^\sigma}\mc{V}^\sigma_i\in\mt{RHS}$ for
some $n$-deep strategy $\sigma\in\Sigma_n(x)$
and set $\mc{V}^\sigma_i\subseteq\mc{W}^\sigma_i$ for each $i\le l^\sigma$.
Construct an $(n+1)$-strategy $\sigma'$ with $\sigma'_i:=\sigma_i$ for
all $i\le n$, $\sigma'_{n+1}(\chi^\sigma_i)=\chi^\sigma_i$
for each $i\le m^\sigma$,
and $\sigma'_{n+1}(\rho^\sigma_ix^\sigma_i):=\bigcup\mc{V}^\sigma_i$
for each $i\le m^\sigma$.
To show $\sigma'\in\Sigma_{n+1}(x)$,
we already have that $\{\sigma'_i\}_{i\le n}$ as it is defined
to be $\sigma\in\Sigma_n(x)$, so we are left to check that $\sigma'_{n+1}$
satisfies the conditions in the definition of a strategy
(Definition~\ref{def:strategy}).
It preserves completed plays.
Finally we show
$\bigcup\mc{V}^\sigma_i=\{\rho^\sigma_ix^\sigma_i u\mid u\in U\}$
for some $U\in c^*(x^\sigma_i)$ (for all $i\le l^\sigma$).
Let $\mc{V}^\sigma_i=\{V_{ij}^\sigma\}_{j\in J}$,
we have some $U^\sigma_{ij}\in c^*(x^\sigma_i)$
with $\{\rho^\sigma_ix^\sigma_iu\mid u\in U^\sigma_{ij}\}=V^\sigma_{ij}$.
Because $c^*(x^\sigma_i)$ is closed under finite unions
we have $U:=\ds\bigcup_{j\in J}U^\sigma_{ij}\in c^*(x^\sigma_i)$,
so:
\begin{align*}
\tts\bigcup\mc{V}^\sigma_i
=\ds\bigcup_{j\in J}V^\sigma_{ij}
&=\bigcup_{j\in J}\{\rho^\sigma_ix^\sigma_iu\mid u\in U^\sigma_{ij}\}
\\
&=\{\rho^\sigma_ix^\sigma_iu\mid j\in J,\, u\in U^\sigma_{ij}\}
\\
&=\{\rho^\sigma_ix^\sigma_iu\mid u\in \bigcup_{j\in J}U^\sigma_{ij}\}
\\
&=\{\rho^\sigma_ix^\sigma_iu\mid u\in U\}
\end{align*}

Now we consider the case $T=D$.
We follow the same as for $T=P$.
Adopting a slightly more compact notation,
for all $\sigma\in\Sigma_n(x)$:
assume
${\sf plays}^\sigma_n(x)=
[\chi^\sigma_i\mapsto p^\sigma_i,\rho^\sigma_jx^\sigma_j\mapsto q^\sigma_j]
_{i\in I^\sigma,j\in J^\sigma}$
where $I^\sigma$ and $J^\sigma$ are two disjoint finite sets,
$\chi^\sigma_i\in (XA)^{<n}XB$, $\rho^\sigma_jx^\sigma_j\in (XA)^nX$,
and $p^\sigma_i,q^\sigma_j\in(0,1]$
are positive probabilities, such that
$\ds\sum_{i\in I}p_i+\sum_{j\in J}q_j=1$.
We thus have
\begin{align*}
    \mt{(LHS)}
    &=\{\mu^D\circ D(\sigma_{n+1})(
    [\chi^\sigma_i\mapsto p^\sigma_i,\rho^\sigma_jx^\sigma_j\mapsto q^\sigma_j]
    _{i\in I^\sigma,j\in J^\sigma})
    \mid \sigma\in\Sigma_{n+1}(x)\}
    \\
    &=\{\mu^D(
    [[\chi^\sigma_i\mapsto 1]\mapsto p^\sigma_i,\sigma_{n+1}(\rho^\sigma_jx^\sigma_j)\mapsto q^\sigma_j]
    _{i\in I^\sigma,j\in J^\sigma}
    )\mid
    \sigma\in\Sigma_{n+1}(x)\}
\end{align*}
Let
$W^\sigma_j=\{[\rho^\sigma_jx^\sigma_ju\mapsto\varphi(u)]_{u\in{\sf supp}\,\varphi}
\mid\varphi\in c^*(x^\sigma_j)\}$.
\begin{align*}
    &
    P\mu^D\circ\delta^{PD}\circ D(\lambda_n)\circ DH^n_X(c^*)
        ({\sf plays}^\sigma_n(x))
    \\
    &=
    P\mu^D\circ\delta^{PD}\circ D(\lambda_n)\circ DH^n_X(c^*)
    ([\chi^\sigma_i\mapsto p^\sigma_i,\rho^\sigma_jx^\sigma_j\mapsto q^\sigma_j]
    _{i\in I^\sigma,j\in J^\sigma}
    )
    \\
    &=
    P\mu^D\circ\delta^{PD}\circ D(\lambda_n)
    ([\chi^\sigma_i\mapsto p^\sigma_i,\rho^\sigma_jx^\sigma_jc^*(x^\sigma_j)\mapsto q^\sigma_j]
    _{i\in I^\sigma,j\in J^\sigma}
    )
    \\
    &=
    P\mu^D\circ\delta^{PD}
    ([\{[\chi^\sigma_i\mapsto 1]\}\mapsto p^\sigma_i,
    W^\sigma_j\mapsto q^\sigma_j]_{i\in I^\sigma,j\in J^\sigma}
    )
    \\
    &=
    P\mu^D(
    \{\mu^D[\Phi^\sigma_i\mapsto p_i^\sigma,\Phi_j^\sigma\mapsto q^\sigma_j]_{i\in I^\sigma,j\in J^\sigma}
    \mid {\sf supp}\,\Phi^\sigma_i\subseteq\{[\chi^\sigma_i\mapsto 1]\},\,
        {\sf supp}\,\Phi^\sigma_j\subseteq W^\sigma_j\}
    )
    \\
    &=
    P\mu^D(
    \{\mu^D[[[\chi^\sigma_i\mapsto 1]\mapsto 1]\mapsto p_i^\sigma,\Phi_j^\sigma\mapsto q^\sigma_j]_{i\in I^\sigma,j\in J^\sigma}
    \mid
        {\sf supp}\,\Phi^\sigma_j\subseteq W^\sigma_j\}
    )
    \\
    &=
    \{\mu^D\mu^D[[[\chi^\sigma_i\mapsto 1]\mapsto 1]\mapsto p_i^\sigma,\Phi_j^\sigma\mapsto q^\sigma_j]
        _{i\in I^\sigma,j\in J^\sigma}
    \mid
        {\sf supp}\,\Phi^\sigma_j\subseteq W^\sigma_j\}
\end{align*}
As is common when notating distributions, we rewrite by leaving
the $\mu^D$'s implicit.
We can also simplify because complete plays and incomplete plays
are disjoint.
\[
    \mt{LHS}=\{
    [\chi^\sigma_i\mapsto p^\sigma_i,\sigma_{n+1}(\rho^\sigma_jx^\sigma_j)\mapsto q^\sigma_j]
_{i\in I^\sigma,j\in J^\sigma}\mid\sigma\in\Sigma_{n+1}(x)\}
\]
\[ \mt{RHS}=\{[\chi^\sigma_i\mapsto p^\sigma_i,\Phi^\sigma_j\mapsto
q^\sigma_j]_{i\in I^\sigma,j\in J^\sigma}\mid \sigma\in\Sigma_{n}(x),\,
{\sf supp}\,\Phi^\sigma_j\subseteq W^\sigma_j\}\]

$(\subseteq)$ Assume some $\sigma\in\Sigma_{n+1}$. Using the same trick
as in the $\widetilde{PQ}$ case, we have that $\sigma\in\Sigma_{c,n}(x)$.
We can choose $\Phi^\sigma_j:=[\sigma_{n+1}(x^\sigma_j)\mapsto 1]$.
We first verify that ${\sf supp}\,\Phi^\sigma_j\subseteq W^\sigma_j$,
which involves checking $\{\sigma_{n+1}(x^\sigma_j)\}\subseteq W^\sigma_j$,
this follows from $\sigma_{n+1}(x^\sigma_j)\in c^*(x^\sigma_j)$
(which holds by def. of strategy).

$(\supseteq)$ Suppose $\sigma\in\Sigma_n(x)$
and $\Phi^\sigma_j\in DDH^{n+1}_X(X)$ such that
${\sf supp}\,\Phi^\sigma_j\subseteq W^\sigma_j$.
We can extend $\sigma$ by choosing
$\sigma_{n+1}(\rho^\sigma_jx^\sigma_j):=\mu^D(\Phi^\sigma_j)$.
This results in a element of the LHS (as again, by convexity,
this element must exist). It is equal to the original because
$\rho^\sigma_jx^\sigma_ju \mapsto \mu^D(\Phi^\sigma_j)(\rho^\sigma_jx^\sigma_ju)\cdot q^\sigma_j$
by both elements.

\end{document}